\documentclass[lettersize,journal]{IEEEtran}
\IEEEoverridecommandlockouts

\usepackage[utf8]{inputenc} 
\usepackage{graphicx}
\usepackage[T1]{fontenc} 
\usepackage{lmodern}
\usepackage{tabu,enumerate}
\usepackage{array,multirow,makecell}
\usepackage{xcolor}
\usepackage{siunitx}
\usepackage{amsmath, amsfonts, amssymb, bm}
\usepackage{mathtools}
\usepackage{pdfpages}
\usepackage{hhline}
\usepackage{comment}
\usepackage{framed}
\usepackage{url}
\usepackage{dsfont}
\usepackage{algorithm}
\usepackage[noend]{algpseudocode}
\usepackage{bbold}
\usepackage{tikz}
\usepackage{subcaption}

\usepackage{titlesec}

\usepackage{ifthen}

\usepackage{amsthm}
\usepackage{multirow}

\newtheorem{definition*}{Definition}
\newtheorem{lemma*}{Lemma}
\newtheorem{theorem}{Theorem}
\newtheorem{corollary*}{Corollary}
\usepackage{lipsum}
\usepackage{comment}
\usepackage{ulem}
\usepackage{parskip} 

\DeclareMathOperator{\rank}{rank} 

\definecolor{brightpink}{rgb}{1.0, 0.0, 0.5}

\begin{document}
\def\viewPrev{0} 

\newcommand{\replace}[2]{\if\viewPrev1{\color{blue}\sout{#1}#2}\else\if\viewPrev2{\color{blue}#2}\else#2\fi\fi}

\newdimen\origiwspc%
\newdimen\origiwstr%
\origiwspc=\fontdimen2\font
\origiwstr=\fontdimen3\font
\title{ Least-squares methods for nonnegative matrix factorization over rational functions
\thanks{This work was supported by the Fonds de la Recherche Scientifique - FNRS and the Fonds Wetenschappelijk Onderzoek - Vlaanderen under EOS Project no 30468160, and by the Leuven Institute for Artificial Intelligence (Leuven.ai)}}
\fontdimen2\font=\origiwspc%
\fontdimen3\font=\origiwstr

\author{\IEEEauthorblockN{Cécile Hautecoeur,}
\and
\IEEEauthorblockN{Lieven De Lathauwer,}
\and\IEEEauthorblockN{Nicolas Gillis,}
\and\IEEEauthorblockN{François Glineur}
}

\newcommand{\Rplus}{\mathbb{R}_+}
\newcommand{\Rp}[2]{\Rplus^{#1 \times #2}}
\renewcommand{\F}{\mathcal{F}}
\renewcommand{\P}{\mathcal{R}}
\newcommand{\projP}[1]{\left[#1\right]_{\P}}

\newcommand{\pB}{\tilde g}

\newcommand{\Mij}[3]{#1_{#2#3}}
\newcommand{\X}{X}
\renewcommand{\Xi}[2]{\Mij{\X}{#1}{#2}}
\renewcommand{\A}{\mathcal{A}}
\newcommand{\Ai}[1]{A_{:#1}}
\newcommand{\Y}{\mathcal{Y}}
\newcommand{\Yi}[1]{Y_{:#1}}

\maketitle

\begin{abstract}
Nonnegative Matrix Factorization (NMF) models are widely used to recover linearly mixed nonnegative data. 
When the data is made of samplings of continuous signals, the factors in NMF can be constrained to be samples of nonnegative rational functions, which allow fairly general models; this is referred to as NMF using rational functions (R-NMF). 
We first show that, under mild assumptions, R-NMF has an essentially unique factorization unlike NMF, which is crucial in applications where ground-truth factors need to be recovered such as  blind source separation problems. Then we present different approaches to solve R-NMF
: the R-HANLS, R-ANLS and R-NLS methods. From our tests, no method significantly outperforms the others, and a trade-off should be done between time and accuracy. Indeed, R-HANLS is fast and accurate for large problems, while R-ANLS is more accurate, but also more resources demanding, both in time and memory. R-NLS is very accurate but only for small problems. Moreover, we show that R-NMF outperforms NMF in various tasks including the recovery of semi-synthetic continuous signals, and a classification problem of real hyperspectral signals.

\end{abstract}
\begin{IEEEkeywords}
nonnegative matrix factorization, block-coordinate-descent, sampled signals, nonlinear least squares, nonnegative rational functions, projection
\end{IEEEkeywords}

\section{Introduction}
Linear dimension reduction techniques are simple but powerful methods to reduce the size of a dataset while extracting meaningful information and filtering noise. 
When the data is nonnegative, it is common to use the Nonnegative Matrix Factorization (NMF). 
In NMF, the nonnegative input data matrix $Y$ is approximated by the product of two nonnegative matrices, $A$ and $X$, such that 
$Y\simeq AX^\top$. The number $r$ of columns of these two matrices must be much smaller than the dimensions of the input matrix, leading to a compressed representation. 
This allows the description of each column of $Y$ as a nonnegative weighted sum  of $r$ characteristic nonnegative elements, the columns of $A$~\cite{lee1999learning}. 

Nonnegativity constraints occur naturally in many situations, e.g., when recording intensities, occurrences, frequencies, proportions, and  probabilities. 
Imposing nonnegativity in the factorization leads to more meaningful decompositions: 
the basis formed by the column of $A$ can be interpreted in the same way as the data, while the input matrix $Y$ is reconstructed using only 
 additive combinations of these basis elements, which leads to a part-based representation~\cite{lee1999learning}. 
This explains the popularity of NMF in various fields such as image processing, text mining, blind source separation, and microarray data analysis; see~\cite{cichocki2009nonnegative, gillis2020nonnegative} and the references therein.

To further improve the quality of the factorization and be even less sensitive to noise, other constraints can be considered on factors $A$ and $X$. For example, when the data is smooth, one can consider that the columns of $A$ are discretizations of continuous nonnegative functions like polynomials~\cite{debals2017nonnegative}, splines~\cite{backenroth2018methods,zdunek2014alternating,zdunek2014b},  or mixture(s) of Gaussian radial basis functions~\cite{zdunek2012approximation}. NMF can then be solved in several ways, but an efficient approach is to generalize the Hierarchical Alternating Least Squares (HALS) algorithm~\cite{cichocki2007hierarchical} 
and solve the problem using block-coordinate descent (BCD) with $2r$ blocks: the columns of $A$ and $X$. This algorithm requires to repeatedly project each block on the considered set of nonnegative functions; for example on nonnegative polynomials and splines~\cite{hautecoeurNeuro}.

When the columns of the input matrix $Y$ are samples of nonnegative continuous signals, mostly smooth with possibly some peaks, it makes sense to consider that they are samples of nonnegtive rational functions. Indeed, when the denominator of a rational function is close to zero, it results in a peak in the signal.
In fact, rational functions are able to represent a large range of shapes and curves \cite{ionita2013lagrange}.
NMF over rational functions, R-NMF, has been introduced in~\cite{hautecoeur2021hierarchical}, and is recalled in Section~\ref{sec:rnmf}. In Section~\ref{sec:unique}, we prove that unlike standard NMF, R-NMF is essentially unique under mild conditions, which is very important when the objective is to recover the sources behind data.

In~\cite{hautecoeur2021hierarchical} it is shown that R-NMF leads to better factorization and reconstruction than standard NMF on noisy data. However, 
the set of nonnegative rational functions of fixed degree is not convex, and the projection on it is not easy to compute. Therefore, the problem is solved using an HALS-like approach, named R-HANLS, that uses an approximate projection method. We explore in Section~\ref{sec:projr} other methods to approximately  project on nonnegative rational functions, with the goal to determine whether some methods lead to better projections and/or if some are more adapted for R-NMF, i.e. lead to better factorizations. 

One of the main advantages of HALS for standard NMF is the simplicity of its iterations. 
However, when using rational functions, each iteration is difficult due to the projection. It is thus questionable whether this approach is suitable, 
so we consider other 
block decompositions 
in Section~\ref{sec:methods}, the R-ANLS and R-NLS methods. 
Methods are then analyzed and compared in Section~\ref{sec:compare}, where we find that R-HANLS is suited for large-scale data
, while R-ANLS obtains a more accurate factorization but is slower. 
R-NLS can only be used for very small data. Moreover, R-NMF is more accurate than NMF using polynomials, splines or vectors on various datasets, like semi-synthetic datasets containing mixture of real reflectance signals, and on a real problem, the Indian Pines classification problem. 

\section{NMF using rational functions}
\label{sec:rnmf}
Consider an input data matrix $Y \in \mathbb{R}^{m\times n}$, containing in each of its columns the samples of a continuous signal taken in $m$ known discretization points $\bm\tau = \{\tau_i\}_{i=1}^m \subset \mathbb{R}$, the sampling points need not to be taken equidistantly. Let $T$ be the interval on which $\bm \tau$ is defined: $T = [\tau_{\min},\tau_{\max}]$, and $\mathcal{F}^{\bm d,T}$ be the set of rational functions of degree $\bm d$ nonnegative on $T$. The goal of R-NMF is to approximate the columns of $Y$ 
 as a nonnegative linear combination of $r$ functions in $\mathcal{F}^{\bm d,T}$. However, 
as the input signals are known only at points $\bm \tau$, to evaluate the quality of the factorization, we focus on the discretization of $\mathcal{F}^{\bm d,T}$ on $\bm\tau: \ \mathcal{R}^{\bm d,T}_{\bm \tau}
=\{f(\bm \tau)|f\in \mathcal{F}^{\bm d,T}\} \subset \mathbb{R}_+^m$, and use the Frobenius norm $\|\cdot \|_F$ of the reconstruction error of $Y$ as objective.

\begin{definition*}[\textbf{R-NMF}] \label{def:f-nmf}  
Given an input matrix $Y\in \mathbb{R}^{m\times n}$, discretization points $\bm \tau \in \mathbb{R}^m$, the set $\P_{\bm \tau}^{\bm d,T}$ of rational functions of degree $\bm d$ nonnegative on $T$ and evaluated on $\bm \tau$, and a factorization rank $r\geq 1$. R-NMF aims to compute a nonnegative matrix $A\in \mathbb{R}^{m\times r}_+$ containing elements of $\P_{\bm \tau}^{\bm d,T}$ in each of its columns, i.e.\@ $A_{:j}\in \P_{\bm \tau}^{\bm d,T} \ \forall j$, 
and a nonnegative matrix $\X\in \Rp{n}{r}$ solving 
\begin{equation} \label{problem}\min_{\Ai{j}\in \P_{\bm \tau}^{\bm d,T}, \X\in \Rp{n}{r}} \sum_{i=1}^n \Big\|\Yi{i} - \sum_{j=1}^r \Ai{j}\Xi{i}{j} \Big\|^2_F.\end{equation}
\end{definition*}
The choice of rational functions is motivated by
their ability to represent a large range of shapes and their utility in applications; 
they generalize polynomials or splines~\cite{trefethen2019approximation}, and they represent the natural way of modeling linear dynamical systems in the frequency domain~\cite{ionita2013lagrange}.  
A rational function is defined as the ratio of two polynomials: $f(t) = \frac{h(t)}{g(t)}$. Throughout this work, we consider univariate rational functions with fixed degree $\bm d=(d_1,d_2)$, so that $h$ is of degree $d_1$ and $g$ of degree $d_2$. 
As the degree is fixed, the set of rational functions is not a vector space (it is easy to check that $\frac{1}{x}+\frac{1}{x+1}$ is of degree $(1,2)$ and not $(1,1)$).  

Nevertheless, this set can be parametrized. Indeed, a rational function nonnegative on a fixed interval can be described as a ratio of two polynomials nonnegative on the same interval~\cite{jibetean2006global}, and nonnegative polynomials can be parametrized using sums of squares~\cite{powers2000polynomials}. Moreover, as it is often undesirable for factors to tend to infinity, the denominator is imposed to be nonzero in the considered interval. More details are presented in \cite{hautecoeur2021hierarchical}. For example a rational function of degree $\bm d=(2d_1',2d_2')$ nonnegative on $[-1,1]$ can be written as: 

\begin{equation}
    \label{eq:ratio_continu} f(t) = \frac{h_1^2(t) + (1-t^2)h_2^2(t)}{g_1^2(t) +(1-t^2)g_2^2(t)+\epsilon}
\end{equation}
with $h_1,h_2,g_1,g_2$ polynomials of degree $d_1',d_1'-1, d_2', d_2'-1$ respectively, and $\epsilon$ prevents the denominator from going to 0. To evaluate $r$ on points $\bm \tau$, we use the Vandermonde-like matrix for the chosen basis of polynomials, $V^d_{\bm \tau}$ (in our case, the Chebyshev basis). Using the coefficients $(\bm h_1, \bm h_2, \bm g_1, \bm g_2) \in \mathbb{R}^{d_1'+1}\hspace*{-0.05cm}\times\hspace*{-0.05cm} \mathbb{R}^{d_1'}\hspace*{-0.05cm}\times\hspace*{-0.05cm} \mathbb{R}^{d_2'+1}\hspace*{-0.05cm}\times\hspace*{-0.05cm} \mathbb{R}^{d_2'} $ we have
\begin{equation}  \label{eq:ratio} \begin{matrix} f_{\bm\tau}(\bm{h_1},  \bm{h_2}, \\ \qquad \ \bm{g_1}, \bm{g_2}) \end{matrix}
=  \frac{(V^{d_1'}_{\bm \tau}\bm{h_1})^2 + (1-\bm \tau^2)\cdot(V^{d_1'-1}_{\bm\tau}\bm{h_2})^2}{(V^{d_2'}_{\bm \tau}\bm{g_1})^2 + (1-\bm \tau^2)\cdot(V^{g_2'-1}_{\bm \tau}\bm{g_2})^2 + \epsilon} \end{equation}

However, this representation is redundant, as multiplying the numerator and the denominator by the same constant leads to the same rational function. 
Therefore, we impose the denominator $g$ to be monic. It can be proven that this condition is equivalent to imposing $\bm{g_1}[d_2'+1] = \frac{\sqrt{8+\bm{g_2}[d_2']^2}}{2}$. 


\section{Uniqueness}
\label{sec:unique}
In this section, we focus on exact factorizations $Y=AX^\top$. In such a factorization, if the column $A_{:i}$ is scaled by a factor $\alpha_i$ while the column $X_{:i}$ is scaled by factor $\frac{1}{\alpha_i}$, $AX^\top$ remains unchanged. 
Moreover, applying the same permutation to the columns of $A$ and $X$ also keeps $AX^\top$ unchanged.
This defines essentially unique factorizations: 
\begin{definition*}
$Y=AX^\top$ is said to have an \textbf{essentially unique factorization} if 
 all the factorizations of $Y$ can be obtained only using consistent permutations and scalings/counterscaling of columns of $A/X$.
\label{def:essen_unique}
\end{definition*}


As shown in Lemma~\ref{lemma:QA}, a matrix $Y$ with factorization $Y\hspace*{-0.02cm}=\hspace*{-0.02cm}AX^\top\hspace*{-0.05cm}$ admits an infinite number of other factorizations not resulting from permutations and scalings. To have an essentially unique factorization, we must add constraints on $A$ and/or $X$. In NMF, the factors $A$ and $X$ are nonnegative. This constraint 
allows, under certain conditions, for an essentially unique factorization. However, these conditions are quite restrictive, and are not met in general, see~\cite{xiao2019uniq} and~\cite[Chap.~4]{gillis2020nonnegative} and the references therein.  
\begin{lemma*}
\label{lemma:QA}
Let $Y=AX^\top $, with $A\in \mathbb{R}^{m\times r}$, $X\in \mathbb{R}^{n\times r}$, and $\rank(Y) = r$. 
Matrices $A'\in \mathbb{R}^{m\times r}$ and $X'\in \mathbb{R}^{n\times r}$ factorize $Y$ if and only if $A'=AQ$ and $X'^\top=Q^{-1}X^{\top}$ where $Q\in \mathbb{R}^{r\times r}$ is an invertible matrix. 
\end{lemma*}
\vspace*{-0.2cm}
\begin{proof}
We omit the proof, which is quite straightforward. 
\end{proof}
\vspace*{-0.2cm}

If we consider that the columns of matrix $A$ are samples of rational functions, it is possible to prove that the product $AX^\top$ is essentially unique under certain conditions on the rational functions contained in the columns of $A$~\cite{debals2015lowner}. Indeed, at most one column can contain a polynomial and
the poles of all rational functions must be distincts. The number of discretization points must also be greater than twice the sum of the degrees of the rational functions in $A$, for example $m>2r(d_1+d_2)$ in R-NMF.

In R-NMF, the considered rational functions must be nonnegative and of the same degrees. The exact R-NMF problem described below is thus a special case of~\cite{debals2015lowner}. Theorem~\ref{thm:unique} shows that it is possible to ensure that  exact R-NMF is essentially unique with milder conditions on $A$. 
\begin{definition*} \label{def:exact}
\textbf{Exact R-NMF} Given $Y \in \mathbb{R}_+^{m \times n}$, $\bm \tau$, $\P_{\bm \tau}^{\bm d,T}$ and $r$ as in R-NMF, compute, if possible, $A\in \mathbb{R}_+^{m\times r}$ with $A_{:j}\in \P_{\bm \tau}^{\bm d,T}$ for all $j$ and $X\in \mathbb{R}_+^{n\times r}$ such that $Y = AX^\top$.  
\end{definition*} 

Let us introduce some 
lemmas and notations. A rational function $f(t)$ of degree $\bm d = (d_1,d_2)$ can be written as: 
\begin{equation}
    \label{eq:ratio_fun} \hspace*{-0.2cm} f(t) = \frac{\alpha \prod_{i=1}^{d_1}(t-z_i)}{\prod_{j=1}^{d_2}(t-p_j)} \quad z_i,p_j \in \mathbb{C}, z_i\neq p_j  \ \forall i,j,  \alpha \neq 0, 
\end{equation}
with $\mathcal{Z}\hspace*{-0.1cm}=\hspace*{-0.1cm} \{z_i\}_{i=1}^{d_1}$ the zeros of $f(t)$, and $\mathcal{P}\hspace*{-0.1cm}=\hspace*{-0.1cm} \{p_j\}_{j=1}^{d_2}$ its poles, including the complex zeros/poles. 
In case of multiple poles, the poles are considered as distinct. Let $f_1, \ f_2 $ be two rational functions with poles $\mathcal{P}_1\hspace*{-0.1cm}=\hspace*{-0.05cm}\{p_1,p_2,p_3\}$ with $p_1\hspace*{-0.1cm}=\hspace*{-0.05cm}p_2\hspace*{-0.1cm}=\hspace*{-0.05cm}p_3$ and   $\mathcal{P}_2\hspace*{-0.05cm} =\hspace*{-0.05cm} \{p_1,p_2\}$ respectively.
The set of all poles is $\{p_1,p_2,p_3\}$ and the set of unique poles, i.e. poles appearing in exactly one function is $\{p_3\}$.

\begin{lemma*}
\label{lem:unique}
Let $\{f_l\}_{l=1}^{r}$ be a collection of rational functions in the form (\ref{eq:ratio_fun}), with $\mathcal{P}_l = \{p_{lj}\}_{j=1}^{d_2}$ holding the poles of $f_l$ and $\mathcal{Z}_l = \{z_{li}\}_{i=1}^{d_1}$ holding the zeros of $f_l$. 
Let $\mathcal{S} = \{s_k\}_{k=1}^m$ be the set of unique poles, i.e. poles appearing in exactly one function.\\
Then any function $f\hspace*{-0.05cm}=\hspace*{-0.05cm}\sum_l \beta_l f_l$ with $\beta_l\hspace*{-0.05cm}\neq\hspace*{-0.05cm} 0$ has a denominator with degree at least equal to the cardinality of $\mathcal{S}$ (=$m$).
\end{lemma*}
\begin{proof}
\vspace*{-0.3cm}
The function $f$ can be written as:
$$ f(t) = \frac{\sum_l \beta_l \alpha_l \prod_{i=1}^{d_1}(t-z_{li}) \prod_{q \in 
{\mathcal{U}\setminus \mathcal{P}_l}}(t-q)}{\prod_{q \in 
{\mathcal{U}}}(t-q)}. \vspace*{-0.1cm}$$
Let $\mathcal{U}$ be the set of all poles in $\{f_l\}_{l=1}^{r}$. We have $\mathcal{S} \subseteq \mathcal{U}$, and all $s_k$ are therefore potential poles of $f$. Let us check if they can be simplified by the numerator or not.
If $s_k\in \mathcal{S}$ is a pole appearing only in $\mathcal{P}_l$, we have $s_k \in \mathcal{U}\setminus \mathcal{P}_i \ \forall i \neq l$. Therefore, when $t=s_k$, only the $l^\text{th}$ term is non-zero in the numerator.
Moreover, $s_k \notin \mathcal{U}\setminus \mathcal{P}_l$ and $s_k \neq z_{li} \ \forall i$ as $s_k$ is a pole of $f_l$. The numerator is therefore nonzero  when $t = s_k$ and $s_k$ is a pole of $f$. As this is valid for all $s_k\in \mathcal{S}$, rational function $f$ has denominator degree at least equal to the cardinality of $\mathcal{S} = m$.
\end{proof}

\begin{lemma*}
\label{lem:maxdeg}
Let $\{f_l\}_{l=1}^r$ be a collection of $r$ rational functions in form (\ref{eq:ratio_fun}), of degree $\bm d = (d_1,d_2)$, and $\bm \tau =\{\tau_i\}_{i=1}^m$ be a set of distinct discretization points with 
$m>d_1+rd_2$, so that the denominators of functions $f_l$ do not cancel at these points. If there exist a rational function $f^*$ of degree $\bm d$ so that $f^*(\bm\tau) = \sum_{l=1}^r \beta_lf_l(\bm\tau)$, then $f^* = \sum_{l=1}^r \beta_lf_l$.
\end{lemma*}
\begin{proof}
Let $\mathcal{Z}_l =\{z_{li}\}_{i=1}^{d_1}$ and $\mathcal{P}_l = \{p_{lj}\}_{j=1}^{d_2}$ be the zeros and the poles of $f_l$ and $\tilde{\mathcal{Z}} =\{\tilde{z}_{i}\}_{i=1}^{d_1}$ and $\tilde{\mathcal{P}} = \{\tilde{p}_{j}\}_{j=1}^{d_2}$ be the zeros and poles of $f^*$. We have
\begin{align} 
& \hspace*{2cm}f^*(\bm \tau) = \sum_{l=1}^r \beta_l f_l(\bm \tau) \nonumber \\
&\hspace*{-0.05cm}\Leftrightarrow  
\frac{\tilde \alpha \prod_{i=1}^{d_1}(\bm \tau-\tilde z_i)}{\prod_{j=1}^{d_2}(\bm \tau-\tilde p_j)} = 
\frac{\sum_{l} \beta_l \alpha_l \prod_{i}(\bm \tau -z_{li}) \prod_{k\neq l, j}(\bm \tau-p_{kj})}{\prod_{l=1}^r\prod_{j=1}^{d_2}(\bm \tau-p_{lj})} \nonumber\\
&\hspace*{-0.05cm}\Leftrightarrow 
\Big(
\tilde \alpha \prod_{i=1}^{d_1}(\bm \tau-\tilde z_i)\Big) \Big(\prod_{l=1}^r\prod_{j=1}^{d_2}(\bm \tau-p_{lj})\Big) =  \label{pol1}\\
&\hspace*{-0.05cm} \ \Big(\prod_{j=1}^{d_2}(\bm \tau-\tilde p_j)\Big) \Big( \sum_{l=1}^r \beta_l \alpha_l \prod_{i=1}^{d_1}(\bm \tau -z_{li}) \prod_{k\neq l}^r\prod_{j=1}^{d_2}(\bm \tau-p_{kj}) \Big) \hspace*{-0.5cm} \label{pol2} \end{align}

Elements (\ref{pol1}) and (\ref{pol2}) are polynomials of degree at most $d_1+rd_2$, evaluated at discretization points $\bm \tau$. As $\bm \tau$ contains $m$ distinct points with 
$m>d_1+rd_2$, these two polynomials must be equal everywhere. Therefore, $f^* = \sum_{l=1}^r \beta_l f_l$.
\end{proof}

We now present conditions on matrices $A$ and $X$ that imply that the exact R-NMF $AX^\top$ is essentially unique.
\begin{theorem}
Let $A\in \mathbb{R}^{m\times r}$ and $X\in \mathbb{R}^{n\times r}$ be of rank~$r$. Suppose all columns of $A$ are the discretizations of rational functions $A_j$ for $j=1,2,\dots,r$, of degree $(d_1,d_2)$ on $m$ distinct points $\bm \tau = \{\tau_i\}_{i=1}^m$, with 
$m > d_1+rd_2$ and $\bm \tau$ not containing poles of the functions $A_j$.  Suppose that for all sets containing 2 functions or more, they are at least $d_2+1$ unique poles, i.e. poles  appearing in exactly one function. 
Then the exact R-NMF $AX^\top$ is essentially unique. 
\label{thm:unique}
\end{theorem}
\begin{proof} 
Let $A',X'$ be such that $A'X'^\top\hspace*{-0.1cm}= AX^\top\hspace*{-0.1cm}$. As $A,X$ are of rank $r$, we know by Lemma~\ref{lemma:QA} that each column $A'_{:j}$ can be written as a linear combination of the columns of $A$:
$A'_{:j} = \sum_{l=1}^r \beta_l A_{:l} = \sum_{l=1}^r \beta_l A_l(\bm \tau).$
To be valid, $A'_{:j}$ must be the discretization of a rational function of degree $(d_1,d_2)$, we name this function $A'_j$. 
As $m>d_1+rd_2$, by Lemma (\ref{lem:maxdeg}), $A'_j$ must be the linear combination of the rational functions in $A$: $A'_j = \sum_l\beta_l A_l$. 

To avoid the trivial case of permutation and scaling, there must be at least one $A'_{:j}$ that is the combination of two or more columns of $A$. As all sets $\{A_i\}$ containing two functions or more have at least $d_2+1$ unique poles, using Lemma~\ref{lem:unique} we know that $A'_j$ has denominator degree at least $d_2+1$. This is in contradiction with the fact that $A'_j$ is a rational function with degree $(d_1,d_2)$. It is therefore not possible to find a valid and not trivial $A'$ so that $A'X'^\top  = AX^\top $ and the factorization $AX^\top $ is essentially unique.
\end{proof}
\begin{corollary*}
\label{col:unique}
Let $A\in\mathbb{R}^{m\times r}$ and $X\in \mathbb{R}^{n\times r}$ be of rank $r$, with the columns of $A$ obtained through evaluation of rational functions of degree $(d_1,d_2)$ on $m$ distinct points $\bm \tau$, 
with $m\hspace*{-0.1cm}>\hspace*{-0.1cm} d_1\hspace*{-0.1cm}+\hspace*{-0.1cm}rd_2$, and $\bm \tau$ not containing poles of functions in $A$.   If each function  has at least $\big\lceil \frac{d_2+1}{2}\big\rceil$ poles distinct from all other functions, the exact R-NMF 
is essentially unique.
\end{corollary*}
Note that the nonnegativity constraint is not necessary for Theorem \ref{thm:unique} and Corollary \ref{col:unique}. 
Nevertheless, when using representation like (\ref{eq:ratio}), functions $A_j$ does not have real poles on interval $T = [\bm\tau_\text{min},\bm\tau_\text{max}]$, thanks to the $\epsilon$ added to the denominator. This means that in this case condition “$\bm \tau$ not containing poles of functions in $A$” is always met.

\section{Algorithms for R-NMF}\label{sec:methods}

In this section, we present three different block decompositions of R-NMF leading to different algorithms.
\begin{figure}[h!]
\vspace*{-0.3cm}
    \centering
    \begin{tikzpicture}
        \node[rotate=90] at (-0.65,0.75){m};
        \node at (0.25,1.65){n};
        \draw[draw=black] (-0.5,0) rectangle (1,1.5);
        \node at (0.25,0.75) {$AX^\top$};
        \node at (0.25,-0.25) {1 block};
        \node at (1.25,0.75) {$=$};
        \draw[draw=black] (1.5,0) rectangle (2,1.5);
        \node[rotate=90] at (1.35,1.2){m};
        \node at (1.75,1.65) {r};
        \node at (2.85,1.65) {n};
        \node at (3.75,1.35) {r};
        \node at (1.75,0.75) {$A$};
        \draw[draw=black] (2.1,1) rectangle (3.6,1.5);
        \node at (2.85,1.25) {$X^\top$};
        \node at (2.5,-0.25) {2 blocks};
        \node at (3.85,0.75) {$=$};
        \node at (4.5,1) {\small $A_{:1}$};
        \node at (5,1.75){\small ${X_{:1}}^\top$};
        \draw[draw=black] (4.1,0) rectangle (4.2,1.5);
        \draw[draw=black] (4.3,1.4) rectangle (5.8,1.5);
        \node at (5.6,0.75) {$+ \cdots +$};
        \draw[draw=black] (6.3,0) rectangle (6.4,1.5);
        \draw[draw=black] (6.5,1.4) rectangle (8,1.5);
        \node at (6.7,1) {\small $A_{:r}$};
        \node at (7.25,1.75){\small ${X_{:r}}^\top$};
        \node at (6.15,-0.25) {2$r$ blocks};
    \end{tikzpicture}
    \vspace*{-0.5cm}
    \caption{ Illustration of the three block-decomposition.}
    \label{fig:illuBloc}
\end{figure}
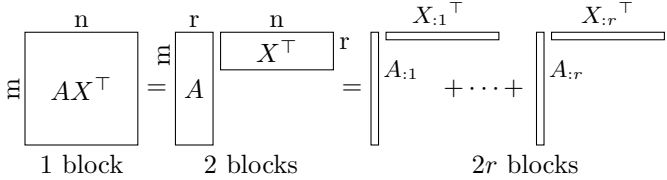

\vspace*{-0.4cm}
\subsection{General Nonlinear Least Squares approach (R-NLS)}
We substitute in~(\ref{problem}) $\Ai{j}$ by $f_{{\bm \tau}_j}$ from equation~(\ref{eq:ratio}), and   $\Xi{i}{j}$ by $C_{ij}^2$  to express R-NMF in an unconstrained way:
\vspace*{-0.2cm}
\begin{equation}
    \label{eq:unconstrained}
    \min_{\substack{\bm{h_{1j}},\bm{h_{2j}},\\ \bm{g_{1j}},\bm{g_{2j}},C}} \sum_{i=1}^n \Big\|\Yi{i} - \sum_{j=1}^r f_{{\bm \tau}_j}(\bm{h_{1j}},\bm{h_{2j}},\bm{g_{1j}},\bm{g_{2j}})C_{ij}^2 \Big\|^2.
\end{equation}
This problem can be solved using a standard nonlinear least squares solver. The same approach for polynomials has been proposed in \cite{debals2017nonnegative}. Note however that in the cited work a compression method is suggested to pre-process the data and reduce the complexity of the problem, but this is not possible in our case because rational function are not linearly parametrizable, that is, they cannot be described using a linear combination of some basis elements. \vspace*{-0.1cm} 
\subsection{Using Alternating Nonlinear Least Squares (R-ANLS)}
Using all-at-once algorithms as R-NLS to solve NMF problems may be computationally costly, especially for large problems. Therefore, many NMF algorithms consider instead alternating schemes \cite{cichocki2007hierarchical},\cite{kim2008nonnegative},\cite{lee1999learning},\cite{lin2007projected}. The problem is then solved by alternating on $A$ and $X$ considering the other matrix as fixed, as sketched in Algorithm~\ref{alg:als}.  As $f_{\bm \tau_j}$ is a nonlinear function, each sub-problem is nonlinear, and this method is called alternating nonlinear least squares. 
\begin{algorithm}[ht]
\caption{Alternating Nonlinear Least Squares}
\label{alg:als}
\begin{algorithmic}
	\Function{R-ANLS}{$Y, \ A,\ X$}
	\While{Stop condition not encountered}:
		\State 
		\vspace*{-0.3cm}\begin{align}\label{eq:upA_als}A \hspace*{-0.1cm} \leftarrow & \underset{\substack{\bm{h_{1j}},\bm{h_{2j}},\\ \bm{g_{1j}},\bm{g_{2j}}}}{\text{argmin}} \hspace*{-0.1cm}\sum_{i=1}^n \hspace*{-0.05cm}\Big\|\Yi{i} \hspace*{-0.1cm} - \hspace*{-0.1cm} \sum_{j=1}^r f_{{\bm \tau}_j}( \bm{h_{1j}},\bm{h_{2j}},\bm{g_{1j}},\bm{g_{2j}})X_{ij} \Big\|^2 \hspace*{-0.2cm}\\ \label{eq:upX_als} \vspace*{-0.2cm}
		X \hspace*{-0.1cm}\leftarrow & \quad \bigg(\ \underset{C\in \mathbb{R}^{n\times r}}{\text{argmin}} \quad  \ \sum_{i=1}^n \Big\|\Yi{i} - \sum_{j=1}^r \Ai{j}C_{ij}^2 \Big\|^2 \bigg)^2\end{align}
	\EndWhile
	\vspace*{-0.1cm}
	\State \textbf{return } $A, \ X$
	\EndFunction
\end{algorithmic}
\end{algorithm}

Problems~(\ref{eq:upA_als}) and  (\ref{eq:upX_als}) are unconstrained and can be solved using a standard nonlinear least squares solver. 
Note that problem (\ref{eq:upX_als}) is separable in $n$ independent  sub-problems, as the rows of $X$, are independent (which is not the case for the rows of $A$): 
\vspace*{-0.1cm}
$$
X_{i:}\leftarrow \ \bigg(\underset{C_{i:}\in \mathbb{R}^r}{\text{argmin}} \quad  \ \Big\|\Yi{i} - \sum_{j=1}^r \Ai{j}C_{ij}^2 \Big\|^2\bigg)^2 \quad \forall i \in \{1,\cdots, n\}.
$$
\subsection{Using Hierarchical Alternating Nonlinear Least Squares (R-HANLS)} 

A popular and effective approach for NMF is  the Hierarchical Alternating Least Squares method (HALS). This method further decomposes the problem in smaller blocks:
the columns of $A$/$X$ are updated successively, considering all the other elements as fixed \cite{cichocki2007hierarchical}; see also~\cite{gillis2012accelerated}. 
 Because of the quadratic structure of the objective function, minimizing~(\ref{problem}) when all variables are fixed except a column of $A$ or $X$ can be done by projecting the unconstrained minimizer on the corresponding feasible region. This region is the set $ \P_{\bm \tau}^{\bm d,T}$ of nonnegative rational functions with fixed degrees (for $A$), or the set $\mathbb{R}_+^{n}$ of nonnegative vectors (for $X$). 
 
 The unconstrained minimizer can easily be found for columns of $A$ and $X$ by cancelling the gradient. Algorithm~\ref{alg:r-hanls} sketches this approach, using $[\cdot]_{S}$ for the projection on set $S$. The projection on $\mathbb{R}_+^n$ is a simple thresholding operation, setting all negative values to 0, while the projection on $ \P_{\bm \tau}^{\bm d,T}$ is not trivial and discussed in the next section. Moreover, equation~(\ref{eq:upX_hals}) is separable: the value of $\Xi{i}{s}$ can be computed independently from $\Xi{j}{s}$, but this is not the case for $\Ai{s}$ in equation~(\ref{eq:upA_hals}), as the projection is not separable unlike the thresholding operation. 
\begin{algorithm}[ht]
\caption{R-HANLS}
\label{alg:r-hanls}
\begin{algorithmic}
	\Function{R-HANLS}{$Y, \ A,\ X$}
	\While{Stop condition not met}
	\For{$\Ai{s} \in A$}
	\vspace*{-0.5cm}
		\State \begin{align}\label{eq:upA_hals} \Ai{s} \leftarrow & \bigg[ \frac{Y\Xi{:}{s}-\sum_{j\neq s}\Ai{j}(\Xi{:}{j})^\top\Xi{:}{s}}{\|\Xi{:}{s}\|^2} \bigg]_{ \P_{\bm \tau}^{\bm d,T}} \hspace*{-0.6cm}
		\end{align}
	\EndFor
	\vspace*{-0.3cm}
	\For{$\Xi{:}{s} \in X$}
	    \begin{align} \label{eq:upX_hals}
		\Xi{:}{s} \leftarrow &  \bigg[ \frac{ Y^\top A_{:s} - \sum_{j\neq s} X_{:j}(\Ai{j})^\top \Ai{s}}{\|\Ai{s}\|^2} \bigg]_{\mathbb{R}^n_+} \end{align}
	\vspace*{-0.2cm}
	\EndFor
	\EndWhile
	\State \textbf{return } $A, \ X$
	\EndFunction
\end{algorithmic}
\end{algorithm}

\section{\hspace*{-0.15cm}Projection on nonnegative rational functions}\label{sec:projr}
As mentioned in Section~\ref{sec:rnmf}, rational functions nonnegative on a fixed interval $T$ can be described as a ratio of two polynomials nonnegative on $T$, with denominator further imposed to be nonzero on $T$. Let $\mathcal{P}^{d}$ be the set of polynomials of degree $d$, $\mathcal{P}_+^{d,T}$ be the set of polynomials of degree $d$ nonnegative on interval $T$, $\mathcal{P}_{++}^{d,T}$ be the set of polynomials of degree $d$ positive on interval $T$, and $\bm z$ be the result of evaluating a function $z(t)$ on discretization points $\bm \tau=\{\tau_i\}_{i=1}^m$, $\bm z = z(\bm \tau)$. Projecting $\bm z$ on rational functions nonnegative on $T$ is therefore equivalent to solving \newpage
\vspace*{-0.85cm}
\begin{equation}
    \min_{h \in \mathcal{P}_+^{d_1,T},\ g \in \mathcal{P}_{++}^{d_2,T}} \big\| \bm z - h(\bm \tau)\big/g(\bm \tau)\big\|^2_2.
    \label{eq:proj_ratio}
\end{equation}
\subsection{Existing approaches to approximate (nonnegative) rational functions} \label{sec:existing}
Solving problem (\ref{eq:proj_ratio}) is not trivial, even when neglecting the nonnegativity constraints. If many works exist in the unconstrained case, most of them consider the infinity norm in (\ref{eq:proj_ratio})~\cite{trefethen2021exponential}, 
and there are very few works imposing nonnegativity: to the best of our knowledge this problem is only addressed in \cite{roh2006discrete,siem2008discrete}, for the infinity norm. 

In the unconstrained case, many works are based on another representation of rational functions, namely the Barycentric representation which is as follows 
\begin{equation}
    f(t) = \sum_{i=1}^d \frac{\omega_i z_i}{t-\alpha_i} \Bigg/ \sum_{i=1}^d \frac{\omega_i}{t-\alpha_i}.
\end{equation}
The advantage of this representation is that the basis used, that is, the sets of $\{\alpha_i\}_{i=1}^d$, can be adapted as the algorithm proceeds to avoid numerical problems at nonsmooth points \cite{filip2018rational}, or Froissart doublets \cite{nakatsukasa2018aaa}. Moreover, when $t\rightarrow \alpha_i$, then $f(t) \rightarrow z_i$, which allows one to optimize only the $\omega_i$. The most common method using this representation is the AAA \cite{nakatsukasa2018aaa}. This method gradually increases the size of the basis by judiciously choosing the $\alpha_i$ points to be added.  It does not seek to optimise a particular norm, but is a good initialization for future optimisation \cite{costa2021aaa,filip2018rational,hokanson2018least,khristenko2021solving}. On the other hand, even if it is not presented as such, one can see Vector Fitting as using the same representation. In this method, the whole basis is chosen at once. Then one optimises iteratively, using at each iteration the poles of the denominator found at previous iteration as new basis \cite{gustavsen1999rational}.

In both methods, once the basis is chosen, the numerator $h$ and denominator $g$ of $f$ are found by optimizing $\|\bm z g(\bm \tau)-h(\bm\tau)\|$ rather than $\|\bm z-h(\bm \tau)/g(\bm \tau)\|$. These methods give good results, but are difficult to use in the context of nonnegative rational functions, because nonnegativity is difficult to express in Barycentric form. 

Many methods try to get rid of the denominator which is difficult to optimise. Thus, \cite{sanathanan1963ratio} but also \cite{loeb1959rational} and \cite{wittmeyer1962rational} have proposed to solve the problem iteratively, using a guess of the denominator, $g^{k-1}$, improved throughout iterations, by solving
\begin{equation}
\label{eq:ratio_num_den}
    (g^k,h^k) = \underset{g\in \mathcal{P}^{d_2},h\in \mathcal{P}^{d_1}}{\text{argmin}} \Big\|\frac{\bm z g(\bm \tau) - h(\bm \tau )}{g^{k-1}(\bm \tau)} \Big\|.
\end{equation}

In the same idea, a special case of the RKFIT algorithm from \cite{berljafa2017rkfit} focuses on finding a good denominator by solving the following problem iteratively:
\begin{equation}
    \label{eq:ratio_rkfit}
    \min_{g^{k}\in \mathcal{P}^{d_2}} \Big\|  \frac{\bm zg^k - h'(g^k; \bm z, g^{k-1})}{g^{k-1}}\Big\| ,
\end{equation}
where $h'(g^k; \bm z, g^{k-1}) = \underset{h}{\text{argmin}} \Big\|\frac{\bm z g^k - h}{g^{k-1}} \Big\|$. 
The problem in $h$ when $g^k$ is fixed has an analytic solution (the solution of a similar problem is presented in Appendix~\ref{sec:implementation}, in the explanation of \textbf{RKFIT+}). This reformulation allows for fewer parameters to be optimised at each iteration. 

When using the infinity norm in (\ref{eq:proj_ratio}), if $g( \tau_i)$ is positive for all $i$, the problem can be rewritten as:
\begin{equation}
    \label{eq:feasibility_prob}
        \min_{h\in \mathcal{P}^{d_1},g(\tau_i)>0,u} u \ \ \text{ s.t.} \ \  \left\{\begin{matrix} z_i g(\tau_i) -h(\tau_i) \leq u g(\tau_i) \\
        h(\tau_i) -z_i g(\tau_i)  \leq u g(\tau_i)
        \end{matrix}\right. .
\end{equation}

If we fix $u$, then the problem is a feasibility problem, and therefore it is possible to perform a bisection search on $u$ to find the solution. This is the method used in \cite{roh2006discrete,siem2008discrete} to solve the problem on nonnegative rational functions. The numerator and the denominator of the rational functions are modeled using Sum Of Squares (SOS), which makes problem (\ref{eq:feasibility_prob}) a SDP feasibility problem for $u$ fixed. 

Finally, using equation (\ref{eq:ratio}), it is possible to see problem (\ref{eq:proj_ratio}) as a nonlinear least squares problem and to solve it using standard  methods \cite{trefethen2021exponential}. 

\subsection{Proposed projection methods} \label{sec:ourProj} Let us present five approaches to solve the projection problem on nonnegative rational functions. Some details of implementation are omitted and presented in Appendix~\ref{sec:implementation} instead, to lighten the text. 

{\textbf{Least Squares:}} Using equation (\ref{eq:ratio}), the projection problem can be rewritten in an unconstrained way and solved using a standard nonlinear least squares solver, as in R-NLS or R-ANLS. This is the approach used in \cite{hautecoeur2021hierarchical}.

{\textbf{Alternating Least Squares:}}
The projection problem can also be divided in two blocks, and solved using a BCD approach. Finding the best possible numerator when the denominator $g$ is fixed is a convex problem on polynomials:
\begin{equation} \label{eq:rkfitNum}
         \underset{h\in \mathcal{P}^{d_1,T}_+}{\text{argmin}} \bigg\|\bm z - \frac{h(\bm\tau)}{g(\bm\tau)}\bigg\|^2. \end{equation}
This problem is described in more details in Appendix~\ref{sec:implementation}. 

When the numerator $h$ is fixed, finding the best denominator is  a challenge as the problem is not convex. Actually this problem is a special case of the projection on rational functions, when the degree of the  numerator is equal to 0. So it can also be solved using nonlinear least squares solvers via equation (\ref{eq:ratio}). As second problem has fewer variables than the original one, we can hope that it will be solved faster.

\begin{algorithm}[ht]
\caption{Alternating Least Squares}
\label{alg:proj_als}
 \textbf{Input:} $\bm z$: signal to approximate, $d_1,d_2$: degree of the numerator/denominator, $\bm\tau$: discretization points, $g$: initial guess of the denominator, tol: tolerance of the algorithm
\begin{algorithmic}[1]
	\Function{Alternating LS}{$\bm z$, $d_1$, $d_2$, $\bm \tau$, $g$, tol}
	\While{$\frac{\text{err}_\text{prev} - \text{err}}{\text{err}} > $tol}
		\State $h = \text{argmin} (\ref{eq:rkfitNum})$
		\State $g = \underset{g\in \mathcal{P}^{d_2,T}_{++}}{\text{argmin}} \|\bm z - {h(\bm\tau)}/{g(\bm\tau)}\|^2.$
		\State $f(\bm \tau) = {h(\bm\tau)}/{g(\bm\tau)}$
		\State err$_\text{prev}$ = err, err = $\| \bm z - f(\bm \tau)\| ^2 $
	\EndWhile
	\State \textbf{return } $f(\bm\tau)$
	\EndFunction
\end{algorithmic}
\end{algorithm}
\vspace*{-0.1cm}

{\textbf{Conic:} }This method is inspired by equation (\ref{eq:ratio_num_den}). From a given estimate of the denominator $\tilde g$, we aim to recover the rational function by optimizing a problem without variables at the denominator. The problem we aim to solve is not the same as in (\ref{eq:ratio_num_den}), and is motivated in Appendix~\ref{sec:implementation}. Indeed, we aim to approximate  $\bm z$ by $f(\bm \tau) = \frac{h(\bm\tau)}{\pB(\bm\tau) + \delta(\bm\tau)}$, with $\pB \in \mathcal{P}_{++}^{d_2,T}$ fixed, by solving 
    \begin{equation}
        \underset{h \in \mathcal{P}^{d_1,T}_+,\delta\in \mathcal{P}^{d_2,T}_+}{\text{argmin}} \bigg\| \frac{\bm z \pB(\bm\tau) +\bm z \delta(\bm\tau) - h(\bm\tau)}{\pB(\bm\tau)}\bigg\|^2.\label{eq:conic}
\end{equation}
    
    Note that the parametrization $f(\bm \tau) = \frac{h(\bm\tau)}{\pB(\bm\tau) + \delta(\bm\tau)}$ allows representing any rational function nonnegative on a fixed interval, and that the  cost function of problem (\ref{eq:conic}) is an upper bound of the  cost function of problem (\ref{eq:proj_ratio}). Moreover, if $\bm z$ is a nonnegative rational function of appropriate degrees, for any $\pB$  it is possible to find $h$ and $d$ such that cost function (\ref{eq:conic}) is equal to zero and $\bm z = f(\bm \tau)$.
    
    The choice of $\pB$ is crucial for this algorithm: the smaller is $\delta$, and therefore the closer is $\pB$ from the denominator of the rational function, the closer are (\ref{eq:conic}) and (\ref{eq:proj_ratio}). Thus problem (\ref{eq:conic}) is solved iteratively, updating $\pB$ as $\pB + \delta$. Note that to avoid to increase $\pB$ indefinitely, it is normalized so that $\pB(\tau_m)=1$ before a new iteration, without loss of generality. This method is sketched in Algorithm~\ref{alg:rkfit}.

     \begin{algorithm}[ht]
\caption{Conic}
\label{alg:rkfit}
 \textbf{Input:} $\bm z$: signal to approximate, $d_1,d_2$: degree of the numerator/denominator, $\bm\tau$: discretization points, $\pB$: initial guess of the denominator, tol: tolerance of the algorithm 
\begin{algorithmic}[1]
	\Function{Conic}{$\bm z$, $d_1$, $d_2$, $\bm \tau$, $\pB$, tol}
	\While{$nb>$tol and $\frac{\text{err}_{\text{prev}} - \text{err}}{\text{err}} > $tol}
		\State $h,\delta = \text{argmin} (\ref{eq:conic})$
		\State $g = \pB + \delta$
		\State $f(\bm \tau) = \frac{h(\bm\tau)}{g(\bm\tau)}$
		\State $nb = \|\pB - g/g(\tau_m)\|^2; \ \pB = g / g(\tau_m)$
		\State $\text{err}_{\text{prev}} = \text{err}; \ \text{err} = \|\bm z - f(\bm \tau)\|^2$
	\EndWhile
	\State \textbf{return } $f(\bm\tau)$
	\EndFunction
\end{algorithmic}
\end{algorithm}

{\textbf{RKFIT+}: }This approach is inspired from the RKFIT method presented in \cite{berljafa2017rkfit}. 
    To find a good denominator, we consider (\ref{eq:conic}) and replace $h(\bm\tau)$ by its best value when $\delta$ and $\pB$ are considered as fixed, without taking into account the nonnegativity constraint. This means that we consider $h'(\pB,\delta,\bm z,\bm\tau) = \text{argmin}_{h\in \mathcal{P}^{d_1}} \big\| \bm z + \frac{\bm z \delta(\bm\tau) - h(\bm\tau)}{\pB(\bm \tau)}\big\|$ instead of $h$. As the nonnegativity constraint is omitted, this problem can be solved analytically using matrix operations (see Appendix~\ref{sec:implementation}). 
    This leads us to the following problem:
    \vspace*{-0cm}
    \begin{equation} \label{eq:rkfit} \underset{\delta \in\mathcal{P}^{d_2,T}_+}{\text{argmin}} \bigg\|
        \bm z + \frac{\bm z \delta(\bm{\tau})}{\pB(\bm{\tau})} -  \frac{h'(\pB,\delta,\bm z,\bm\tau)}{\pB(\bm\tau)} \bigg\|^2 . 
    \end{equation}
    
    To find a good projection on the set of nonnegative rational functions we iterate over instances of problem (\ref{eq:rkfit}). An iterative scheme is useful because problem (\ref{eq:rkfit}) relies on the fixed parameter $\pB$. The pseudo-code of RKFIT+ is quite similar to the one of Conic (Algorithm~\ref{alg:rkfit}). Line 7 is deleted, and lines 3 and 4 are replaced by $g = \text{argmin} (\ref{eq:rkfit})+\pB. $
    Moreover, problem (\ref{eq:rkfitNum}) is solved after the while loop to recover the numerator.

{\textbf{LinProj: }}This approach has been inspired from \cite{roh2006discrete,siem2008discrete}. In this case we consider the infinity norm instead of the squared norm, to express the problem as a bisection search over feasibility problems on polynomials as in (\ref{eq:feasibility_prob}). 
    These feasibility problems can even have linear constraints if we impose $h$ and $g$ to be nonnegative on points $\tau_i \in \bm \tau$ instead of being nonnegative on interval $T$ (this is different from what is done in \cite{roh2006discrete,siem2008discrete}). 
    The feasibility problem is then: 
        \begin{equation} \label{eq:linproj} \min_{h(\tau_i)\geq 0,g(\tau_i)\geq 1} 0 \ \ \  \text{s. t.} \ \ \left\{\begin{matrix}\bm z_ig(\tau_i) - h(\tau_i) \leq ug(\tau_i) \\  h(\tau_i) - \bm z_ig(\tau_i) \leq ug(\tau_i) \end{matrix}\right. \  \forall i , \end{equation}
     and a bisection algorithm is sketched in Algorithm~\ref{alg:lin}. Note  that $ g  $ is prevented from containing values smaller than $1 \ $ at  points $\tau_i$ without loss of generality, to simplify the feasibility problem, preventing [$-ug(\tau_i),ug(\tau_i)$] from being too small.  

    
    \begin{algorithm}[ht]
\caption{LinProj}
\label{alg:lin}
 \textbf{Input:} $\bm z$: signal to approximate, $d_1,d_2$: degree of the numerator/denominator, $\bm\tau$: discretization points, tol: tolerance of the algorithm 
\begin{algorithmic}
	\Function{LinProj}{$\bm z$, $d_1$, $d_2$, $\bm \tau$, tol}
	\State $u_{\max} = \max_{i}  \big\{ \bm z_i - \sum_{s=1}^m \bm z_s/m \big\}; \quad u_{\min} = 0$
	\While{$u_{\max}-u_{\min}$ > tol }
	    \State $u_\text{med} = (u_{\max} + u_{\min})/2$
	    \If{ problem (\ref{eq:linproj}) on $u_\text{med}$ is feasible}
	        \State $u_{\max} = u_\text{med} $
	    \Else
	        \State $u_{\min} =u_\text{med} $
	   \EndIf
	\EndWhile
	\State Find $h,g$ a feasible solution of (\ref{eq:linproj}) using $u_{\max}$
	\State \textbf{return} $\frac{h(\bm\tau)}{g(\bm\tau)}$
	\EndFunction
\end{algorithmic}
\end{algorithm}

\subsection{Comparison of the projection methods}

We now compare those five projection approaches. Algorithms have a tolerance \texttt{tol} of $10^{-8}$.  We consider two sets of inputs:
\begin{itemize}
    \item The signals to project are the discretization of nonnegative rational functions, whose numerator and denominator degrees are $d_1$ and $d_2$, respectively. An exact recovery is thus possible (exact).
    \item The signals to project are the same as in previous case except that we add a Gaussian noise with noise level 20dB (noisy).
\end{itemize} 
Unless specified otherwise, the rational functions have degree $(16,16)$, with $250$ discretization points equally spaced on $[-1,1]$. Fig.~\ref{fig:compProjR} displays the results. The quality of the final projection is computed as the squared norm of the difference between the signal to project and the computed projection, divided by the squared norm of the signal to project. The first observation from this figure is that no method outperforms all others. Indeed, even though RKFIT+ seems quite appropriate for "exact" data, as it obtains the lowest relative error and is among the fastest, it is quite inaccurate for noisy data. On the contrary, Least Squares and Alternating Least Squares provide the best projections on noisy data, but they obtain high errors when there is no noise. When comparing these two approaches, the Least Squares appears to be the best as it is significantly faster and obtain more accurate results. Therefore, we do not consider Alternating Least Squares in what follows. The Linproj generally obtains low relative errors, but sometimes it is unable to find a good candidate when there is noise.  Finally, the Conic approach is not very accurate compared to the others, but it is the fastest.
\hspace*{-0.5cm}
\begin{figure}[h!]
\vspace*{-0.3cm}
\begin{tikzpicture}
\node at (-3.5,-0.1) {Increasing degree};
\node at (-3.3,-6.75) {Increasing number of discretization points};
\node at (-3.5,-0.5) {\small \textbf{Exact data:}};
\node at (-3.5,-7.15) {\small \textbf{Exact data:}};
\node at (-5.2,-2.12){\includegraphics[width = 0.22\textwidth,trim = 3.2cm 0cm 16cm 1.8cm,clip]{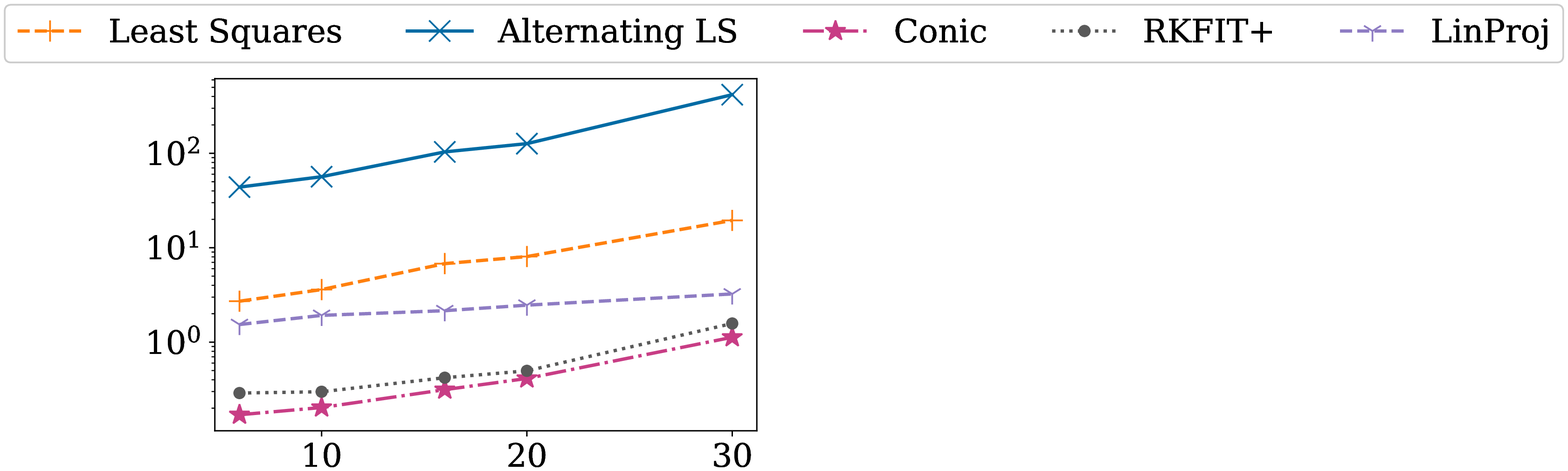}};
\node at (-1.55,-2.1){\includegraphics[width = 0.22\textwidth,trim = 0.3cm 0cm 8cm 0cm,clip]{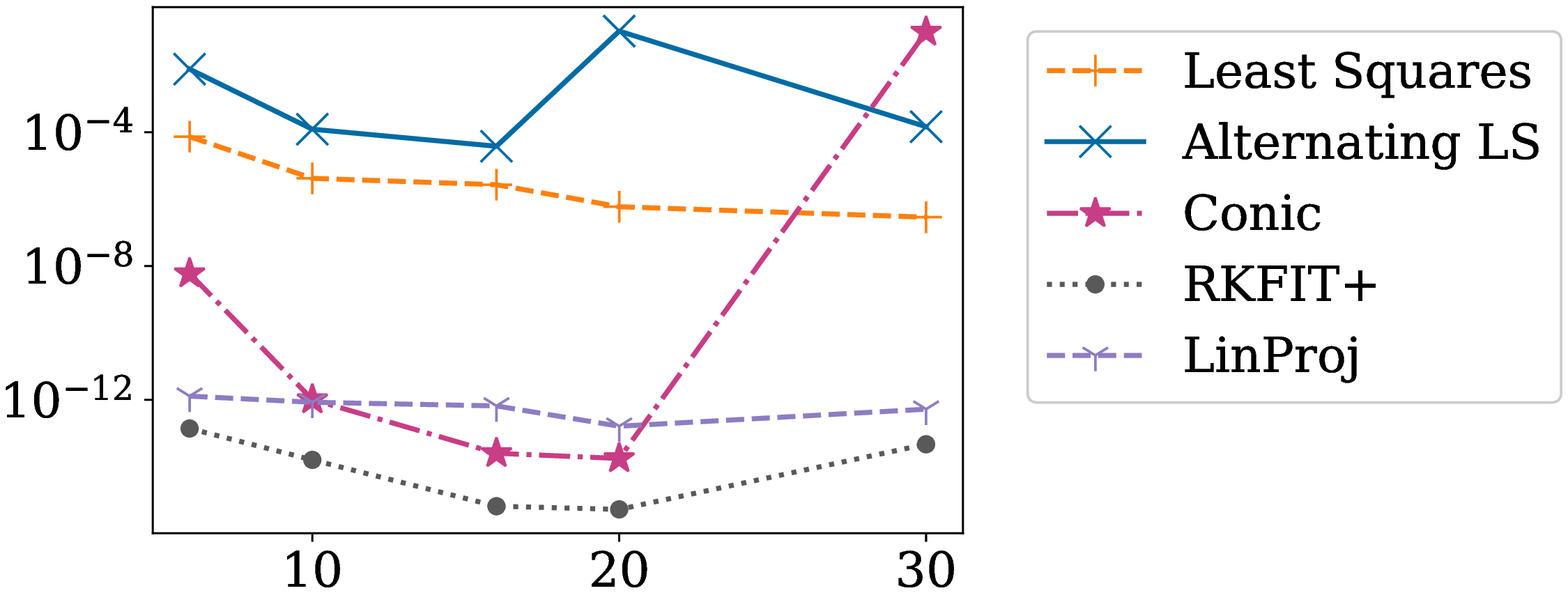}};
\node at (-5.2,-8.79){\includegraphics[width = 0.22\textwidth,trim = 3.2cm 0cm 16cm 1.5cm,clip]{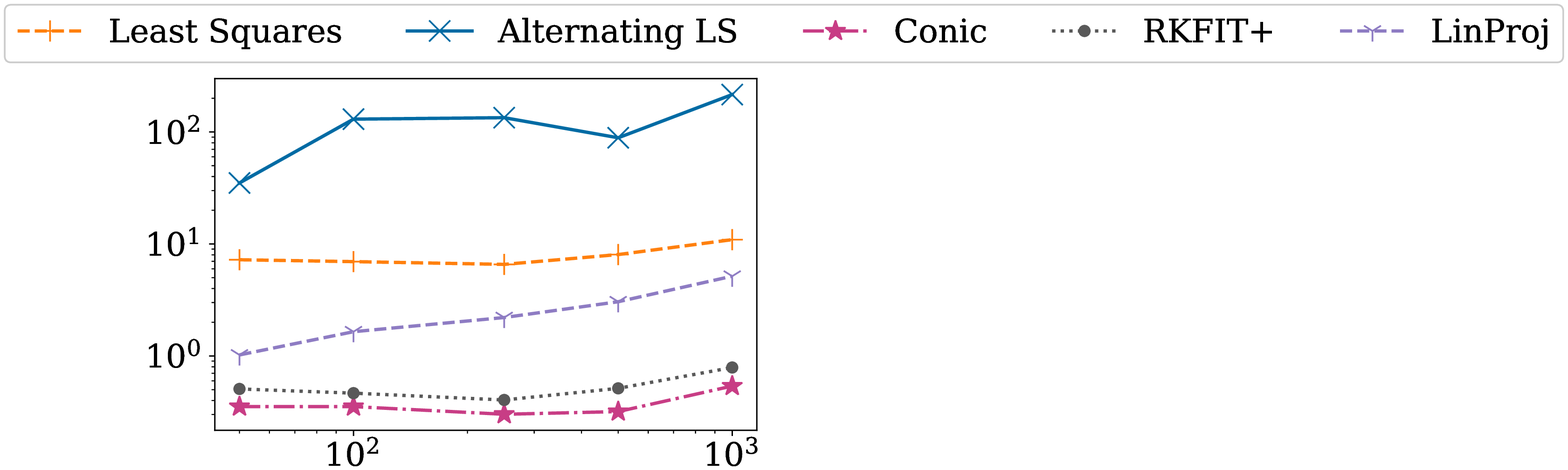}};
\node at (-1.55,-8.8){\includegraphics[width = 0.22\textwidth,trim = 0.4cm 0cm 8.2cm 0cm,clip]{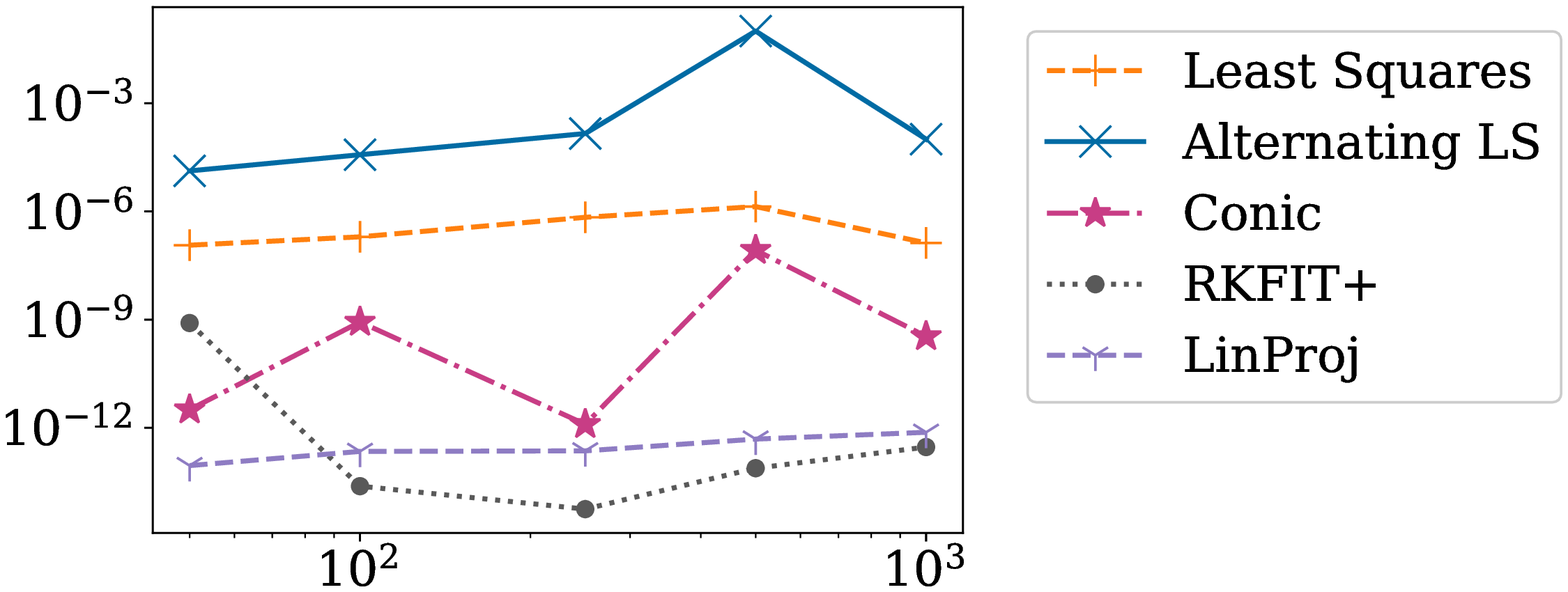}};
\node at (-5.2,-0.8) {\small Times(sec)};
\node at (-5.6,-3.15){\scriptsize degree};
\node at (-1.5,-0.8) {\small Relative error};
\node at (-1.8,-3.18){\scriptsize degree};
\node at (-5.2,-7.45) {\small Times(sec)};
\node at (-4.9,-9.85){\scriptsize discr. points};
\node at (-1.5,-7.45) {\small Relative error};
\node at (-1.2,-9.85){\scriptsize discr. point};

\node at (-3.5,-3.5) {\small \textbf{Noisy data:}};
\node at (-3.5,-10.2) {\small \textbf{Noisy data:}};
\node at (-5.2,-5.2){\includegraphics[width = 0.22\textwidth,trim = 3cm 0cm 16cm 1.8cm,clip]{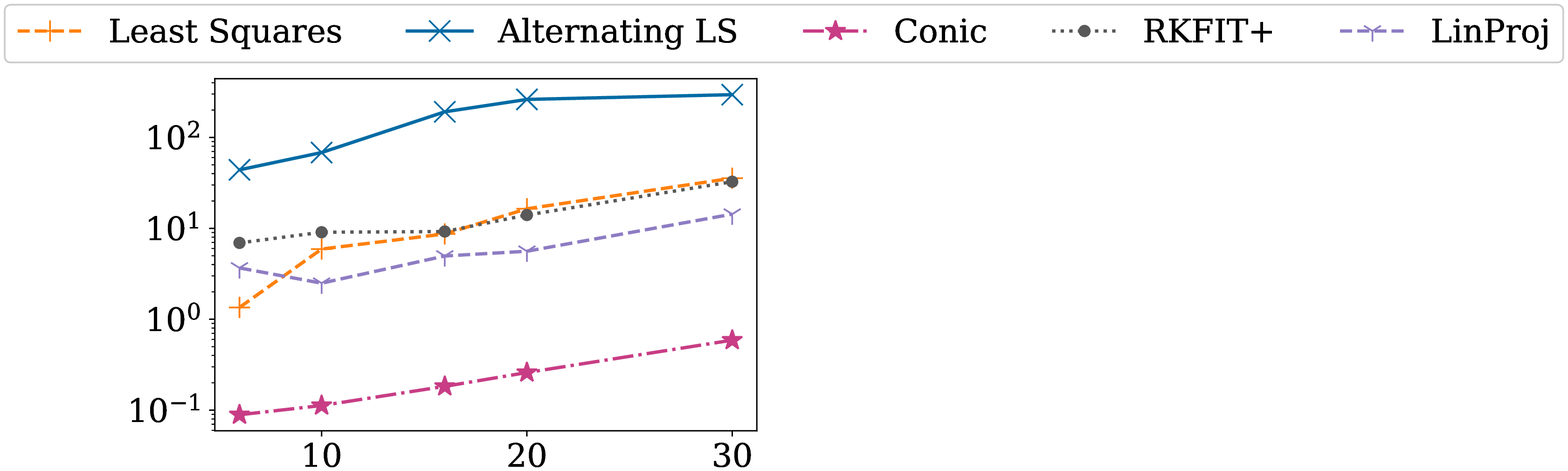}};
\node at (-1.55,-5.18){\includegraphics[width = 0.21\textwidth,trim = 0.3cm 0cm 8cm 0cm,clip]{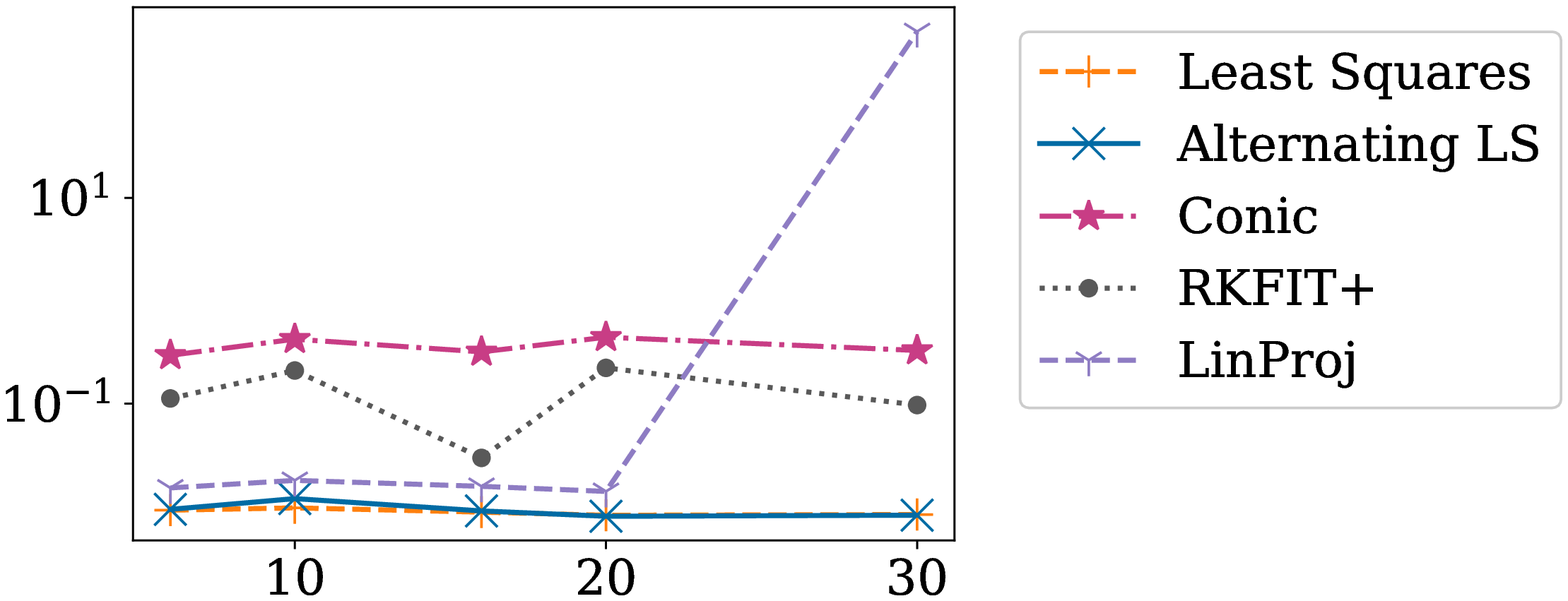}};
\node at (-5.2,-11.7){\includegraphics[width = 0.22\textwidth,trim = 3.2cm 0cm 16cm 1.5cm,clip]{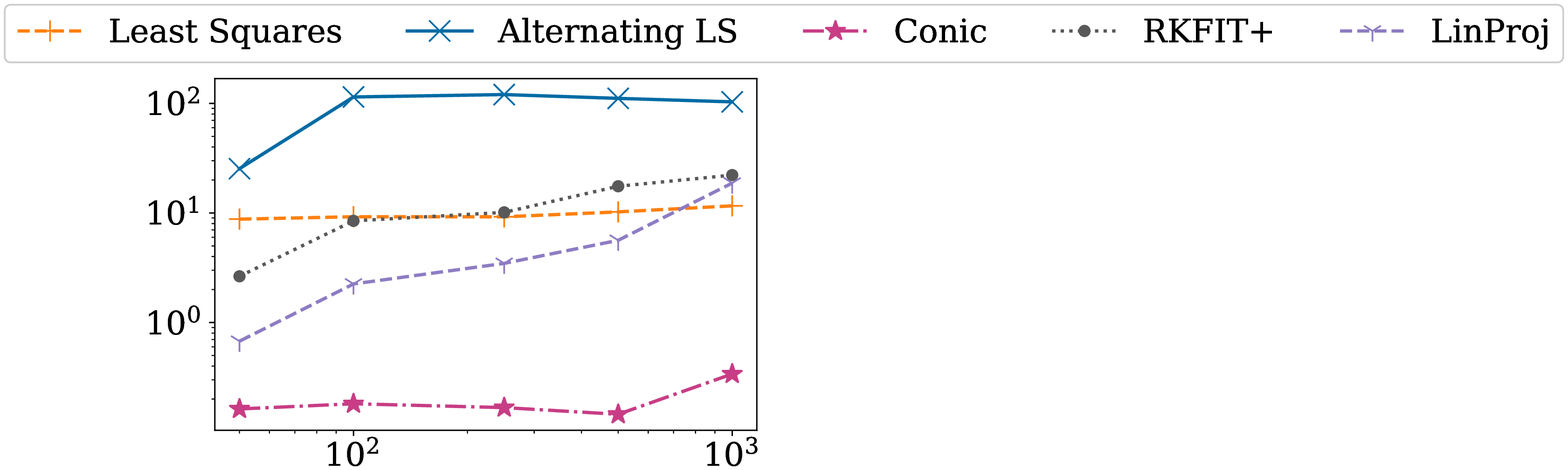}};
\node at (-1.55,-11.75){\includegraphics[width = 0.21\textwidth,trim = 0.3cm 0cm 8.2cm 0cm,clip]{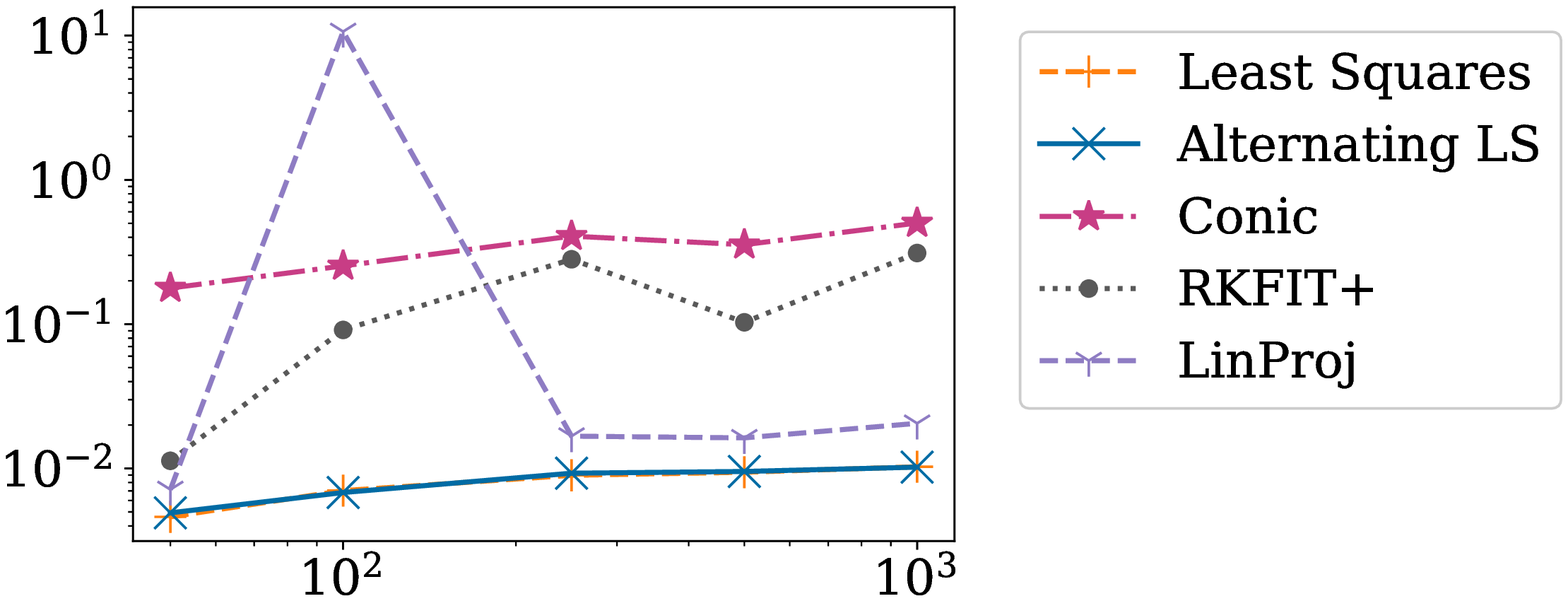}};
\node at (-5.2,-3.85) {\small Times(sec)};
\node at (-5.54,-6.25){\scriptsize degree};
\node at (-1.7,-3.85) {\small Relative error};
\node at (-1.81,-6.25){\scriptsize degree};
\node at (-5.2,-10.42) {\small Times(sec)};
\node at (-4.9,-12.82){\scriptsize discr. points};
\node at (-1.7,-10.45) {\small Relative error};
\node at (-1.2,-12.82){\scriptsize discr. points};

\node at (-3.5,-13.2){\includegraphics[width = 0.48\textwidth, trim = 0cm 9cm 0cm 0cm,clip]{figures/PerfsProjR_Nnoisy.eps}};
\end{tikzpicture}
    
    \caption{Comparison of the projections. The results are averaged over 10 trials. The plots represents the time needed for computations (left) and the relative error (right). 
    }
    \label{fig:compProjR}
\vspace*{-0.3cm}
\end{figure}

We conclude from these experiments that the Least Squares and the RKFIT+ methods seem the more promising projection methods, but they are not always better than the others, and do not outperform them significantly.

\vspace*{-0.1cm}
\section{Performance and Comparison of R-NMF algorithms}
\label{sec:compare}
In this section, we first briefly discuss the computational complexity of the proposed algorithms. 
Then, we compare the R-NMF algorithms
presented in Section~\ref{sec:methods} on purely synthetic datasets to analyze their reconstruction ability and their efficiency. After that, the most promising methods are compared to standard HALS and HALS using polynomials or splines \cite{hautecoeurNeuro} on semi-synthetic datasets. We chose to use  HALS because it is fast and obtains comparable results in terms of accuracy as other  approaches \cite{hautecoeurNeuro}. The methods are also compared on a classification task on a real dataset: the Indian 
Pines dataset\footnote{\url{http://www.ehu.eus/ccwintco/index.php?title=Hyperspectral_Remote_Sensing_Scenes#Indian_Pines}}. 

The least squares solver used for the experimentation is the function \texttt{least\_squares} from python
, with default parameters. The least squares problems are therefore solved using a trust region reflective algorithm \cite{branch1999subspace}. 

\subsection{Algorithmic complexity of the methods}
\label{sec:complex}
Let the following reasonable assumption apply: $r<d<n,m$, where $d$ is the number of degrees of freedom of the used function, e.g.,  $d_1+d_2+1$ for rational functions, the degree plus one for polynomials, and the number of interior knots plus two for splines. The number $r$ is the rank of factorization, $n$ is the number of observations, 
and $m$ is the number of discretization points. Let the complexity of the least squares solver be $ls(k)$ where $k$ is the size of the Jacobian, and $p(k)$ be the complexity of the projections for polynomials and splines, where $k$ is the number of variables to optimize by the algorithm.

We know that an update of HALS for $\X$ has complexity $\mathcal{O}(rmnI)$, where $I$ is the number of iterations. The complexities of HALS using polynomials or splines from \cite{hautecoeurNeuro}, and of R-HANLS using least-squares projection, R-ANLS and R-NLS can also be computed. Their value is summarized in Table~\ref{tab:complex}.
Among HALS methods, R-HANLS is the slowest. Indeed, rational functions are not linearly parametrizable and $m$ appears in the complexity, unlike for polynomials or spline, where $m$ is replaced by $d$ which is significantly lower. 
Nevertheless, R-HANLS is much faster than R-ANLS or R-NLS for large datasets.  
\begin{table}[h]
    \centering
    \begin{tabular}{|c|c|c|}
    \hline
         \small HALS & \small Poly/splines & \small  R-HANLS \\
        \hline
         \small $\mathcal{O}(rmnI)$& \small $\mathcal{O}(rdnI + rp(d^2)I)$ & \small $\mathcal{O}(rmnI + r \ ls(md)I)$  \\
         \hline
    \end{tabular}
    \vspace*{0.3cm}
    
    \begin{tabular}{|c|c|}
    \hline
         \small R-ALS & \small R-LS  \\
        \hline
        \small $\mathcal{O}(rmnI + ls(rdmn)I)$ & \small $\mathcal{O}(ls(rmn(n+d)))$ \\
         \hline
    \end{tabular}
    \vspace*{0.3cm}
    \caption{Computational complexity of the various NMF methods.}
    \label{tab:complex}
\end{table}
\vspace*{-0.5cm}

\subsection{Datasets} 
We use synthetic datasets generated as follows. 
We generate matrix $X \in \mathbb{R}_+^{n\times r}$ randomly, following a Dirichet distribution whose parameters are equal to $\alpha = 1/r$. 
The data provided to the algorithms is $Y = A\X^\top + N$ where $N$ is additive Gaussian noise with known Signal to Noise Ratio (SNR). The matrix $A$ is generated in two ways: 
\begin{itemize}
    \item a "purely synthetic" $A$ which is the discretization of $r$ nonnegative rational functions. The functions are generated as follows. We first create a nonnegative polynomial of degree $d_1$ that is perturbed using a rational function of degree $(1,2)$. This creates a smooth signal with some peaks. The signal is then projected on the set of nonnegative rational functions of degree $(d_1,d_2)$. In this situation, it is therefore possible to find the exact solution of the problem. 
    
    \item a "semi-synthetic"   $A$ whose columns are the real reflectance signals of Adulania, Clinochlore, Hypersthene, Olivine, Spessatine, Andesine, Celestine and Kaolinite evaluated on 414 nonequally spaced points. 
    These signals are showed in Fig.~\ref{fig:refl_sig} (left) and come from the U.S. Geological Survey (USGS) database \cite{USGSdata}. Those signals are not particularly close to rational functions, but they are generally smooth even though they present some peaks. If $r$ is smaller than $8$, we only consider the first $r$ signals in the list.  
\end{itemize}
In all our experiments we impose  methods to have the same number of degrees of freedom (except standard HALS which operates over unstructured nonnegative vectors). This means that if we use rational functions with degree ($d_1$,$d_2$), we use polynomials of degree $d_1+d_2$, and splines of degree 3 with $d_1+d_2-1$ interior knots. Let $A^k$, $X^k$ denote the factors obtained at iteration $k$. 
Accuracy is evaluated trough the relative residue computed as 
\begin{equation}
\frac{\|AX^\top - A^k{X^k}^\top  \|}{\|AX^\top \|}.\label{eq:error}
\end{equation}
Note that this evaluation is performed on $AX^\top $, that is, the data before adding the noise, and therefore the quality is evaluated on data not provided to the algorithm. The stopping criterion of the algorithms is the following:
\begin{equation}sc^k = \frac{\|Y \hspace*{-0.05cm} - \hspace*{-0.05cm} A^{k-1}{X^{k-1}}^\top \| - \|Y \hspace*{-0.05cm}-\hspace*{-0.05cm} A^k{X^k}^\top   \| }{\|Y - A^k{X^k}^\top  \hspace*{-0.1cm} \|} < 10^{-12}\hspace*{-0.1cm}.\label{eq:stopC}\end{equation}
We also impose algorithms to have a maximum running time. Methods based on HALS are limited to 200 seconds, while R-ANLS and R-NLS are limited to 1000 seconds. 
These times have been inspired from Table~\ref{tab:complex}, and selected to be not too important, while allowing the algorithm to converge in most cases, as we will see in the experiments.  

We also report the quality of factorizations by computing the Signal to Interference Ratio (SIR) between the computed $A'$ and the original $A$. The larger the SIR, the closer $A'$ is to $A$. 
As the factors can be 
permuted without loss of generality, we first compute the best permutation 
of $A'$ before computing the SIR. 

In what follows, each test is performed 10 times, using different initializations. To summarize the performance, we compute the minimal and the maximal value obtained for each criterion, and put a marker at the mean value of the criterion. 
If the graph shows the evolution of two criteria with respect to a parameter (like $n$, $m$, $d$ or $r$), only the mean value is presented to improve readability.
We consider that an algorithm converged at iteration $k$ if $\frac{sc^k-sc^{o}}{sc^k}<10^{-3}$ for all $o\geq k$. 
This is used to evaluate the time needed by each algorithm, that it is the time needed to converge. 

\begin{figure}[h!]
\hspace*{-0.1cm}
\begin{tikzpicture}
\node at (-2.5,-1.7){\includegraphics[width = 0.24\textwidth, trim= 0cm 0cm 0cm 0cm, clip]{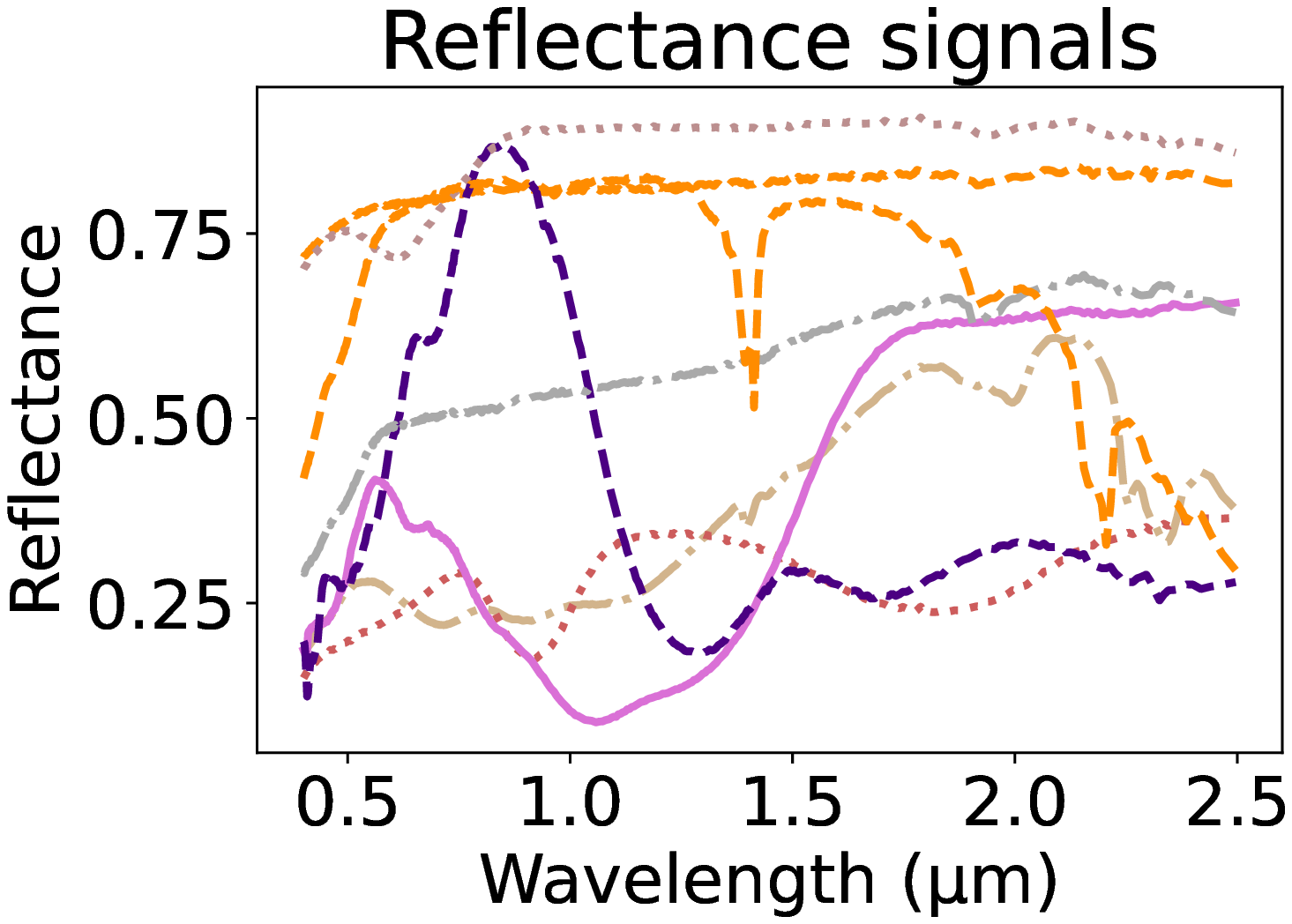}};
\node at (2,-1.7){\includegraphics[width = 0.24\textwidth, trim= 0cm 0cm 0cm 0cm, clip]{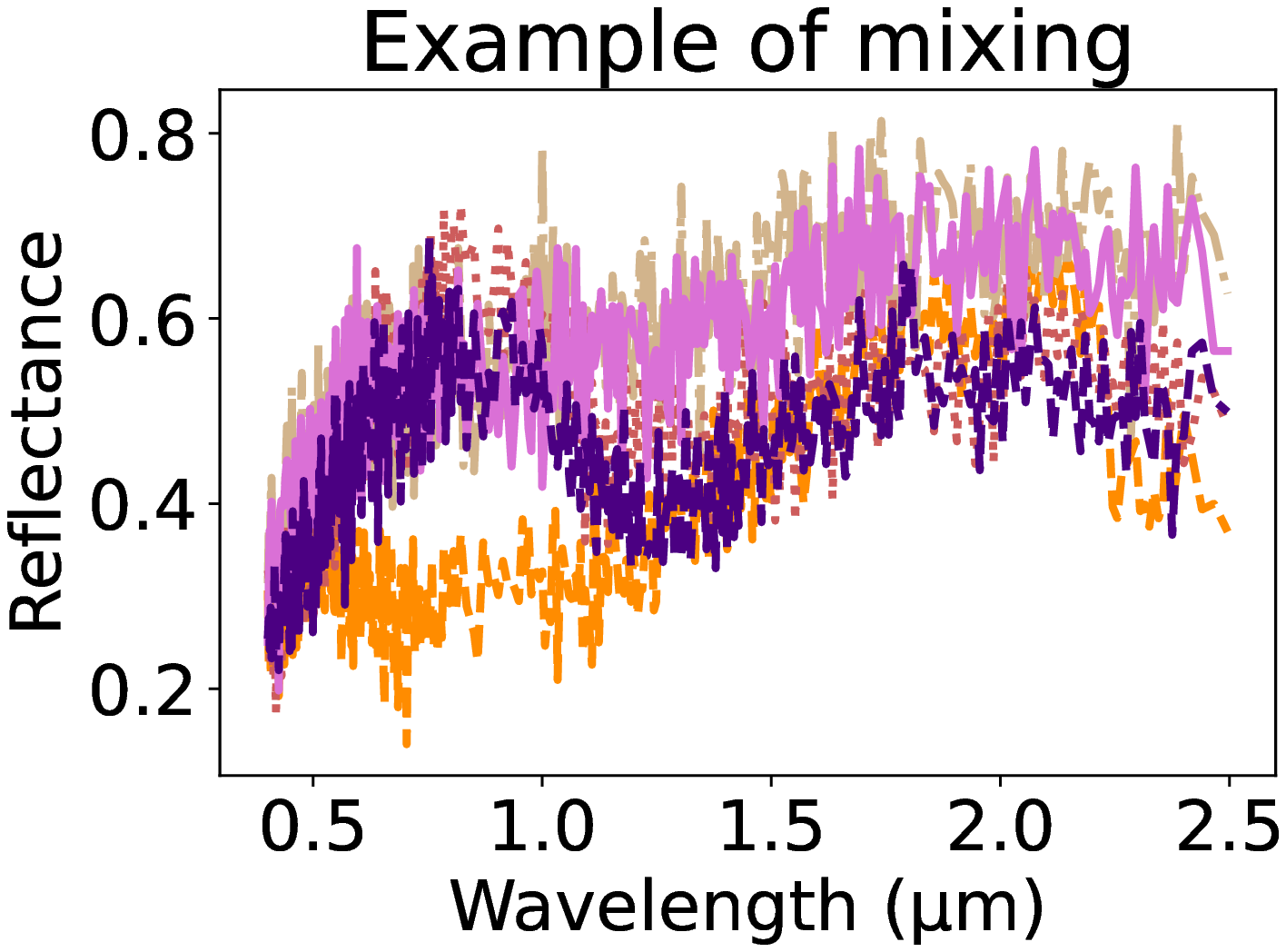}};
\end{tikzpicture}
\caption{Left: Considered real reflectance signals.
    Right: Example of mixing of those signals with noise level 20dB. Each of the five signals is a column of $Y$.}
    \label{fig:refl_sig}
\vspace*{-0.3cm}
\end{figure}

\subsection{Initialization of the projections in R-HANLS}

To get the best out of R-HANLS with the different projections, we use the fact that the last iterates of R-HANLS tend to become close to each other. Therefore, we exploited knowledge from previous iterations, as suggested in \cite{hautecoeur2021hierarchical}: 
\begin{itemize}
    \item Least Squares: use the previous projection as a starting point of the least squares solver.
    \item Conic and RKFIT+: use the previously obtained denominator as first guess. 
    \item LinProj: use a potentially better $u_{\max} = \max_{i} \{ |\bm z_i - f_\text{prev}(\tau_i) | \}$ . 
\end{itemize}

Moreover, the tolerance of the projection methods is decreased progressively from $10^{-2}$ to $10^{-8}$, and Conic and RKFIT+ are limited to one iteration. This leads to accurate results in a reasonable time. Nevertheless, we noted during experiments that using knowledge from previous iterations is particularly beneficial for Least Squares.

\subsection{Purely synthetic dataset}
Let us present the result with or without noise.

\textbf{Case without noise: } 
In this case there is no noise to filter, but it is still interesting to analyze the data and find the factors behind them. By the uniqueness property of rational functions presented in Section~\ref{sec:unique}, we can hope that the methods based on rational functions are able to recover the original signals. We observe in Fig.~\ref{fig:illuNoNoise}  that even though the SIR of methods using rational functions are on average better than the SIR recovered by HALS (which uses any nonnegative vector to represent each column of $A$), this is not always the case, and there is much more variability on the results when using rational functions than when using HALS.  Nevertheless, the best SIR obtained by methods using rational functions are much better than the best SIR obtained when using HALS (except for R-HANLS using LinProj projection). 

Moreover, HALS obtains the best residue, which is expected as it has much more degrees of freedom.  It is therefore difficult to beat HALS in terms of pure data approximation when data is noiseless. Among methods using rational functions, we can see that the LinProj projection is not appropriate; this method is therefore not presented in what follows. The other R-NMF methods have similar performances, except in terms of computation time. Nevertheless, it seems that R-ANLS is the most accurate method in terms of obtained residue, while R-HANLS-based methods are faster. 
\begin{figure}
    \centering
    \includegraphics[width=0.5 \textwidth]{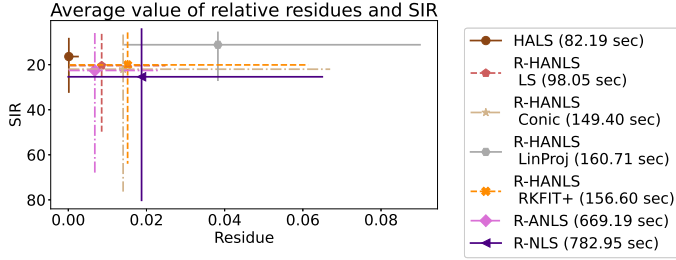}
    \vspace*{-0.5cm}
    \caption{Average performance, with $n=[20,100]$, $d=[6,10]$, $r=[5,10]$. Data is not noisy.}
    \label{fig:illuNoNoise}
\end{figure}

We observe in Fig.~\ref{fig:n_nonoise} that when the number of observations $n$ is small ($n=20$),  R-NLS is able to recover the original signals, as this method obtains a low residue and a high SIR. However, it is unable to do so when the number of observations  increases. We may wonder if this bad result is due to a too tight time constraint, which prevents the algorithm from converging, but even by running the algorithm for 1h (that is, three times longer), the performance did not improve significantly. R-ANLS is the most robust method among methods using rational functions when $n$ changes as its residue is not impacted by this change, unlike other R-NMF methods. 
\begin{figure}\vspace*{-0.3cm}    \centering
    \includegraphics[width=0.5 \textwidth]{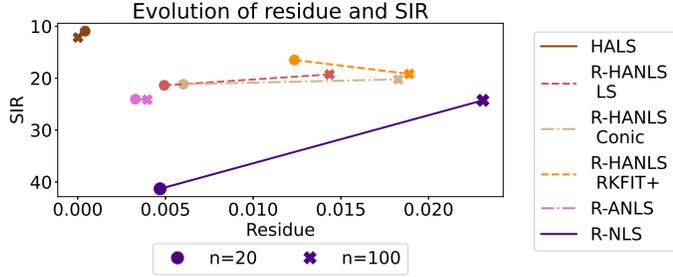}
    \caption{Performance for varying $n$. Data is not noisy. }
    \label{fig:n_nonoise}
    \vspace*{-0.3cm}
\end{figure}

\textbf{Case with noise: }
When noise is added to the dataset, NMF is also useful to filter noise in the data, which can be evaluated through the relative residue (\ref{eq:error}): a low relative residue means a good ability to filter the noise. Fig.~\ref{fig:noise} shows the average results for low and high noise levels. 
We observe that the performance of all algorithms deteriorates when the level of noise increases, as expected. 
Using the Conic or RKFIT+ projections in R-HANLS does not work well when the noise level is high. The noise level has a high impact on the residue of HALS, which means that this method is not good at filtering the noise on the data. However, the quality of the recovered factors is not much impacted by the noise level and stays around 35 dB. R-HANLS LS and R-ANLS obtains the best performances when the noise level is high both in terms of SIR and residue.  We see in Fig.~\ref{fig:n_noise} that increasing $n$,  the number of observations, has a very different impact depending on the used methods: it makes R-NLS perform worse, but it helps the other methods, especially HALS. 
\begin{figure}
    \centering
    \includegraphics[width=0.46 \textwidth]{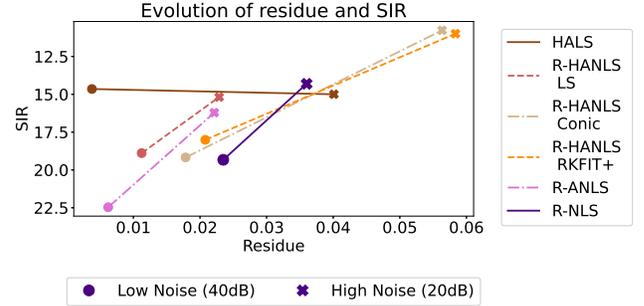}
    \caption{Average performance for varying level of noise.}
    \label{fig:noise}
\end{figure}
\begin{figure}
\vspace*{-0.2cm}
    \centering
    \includegraphics[width=0.46 \textwidth]{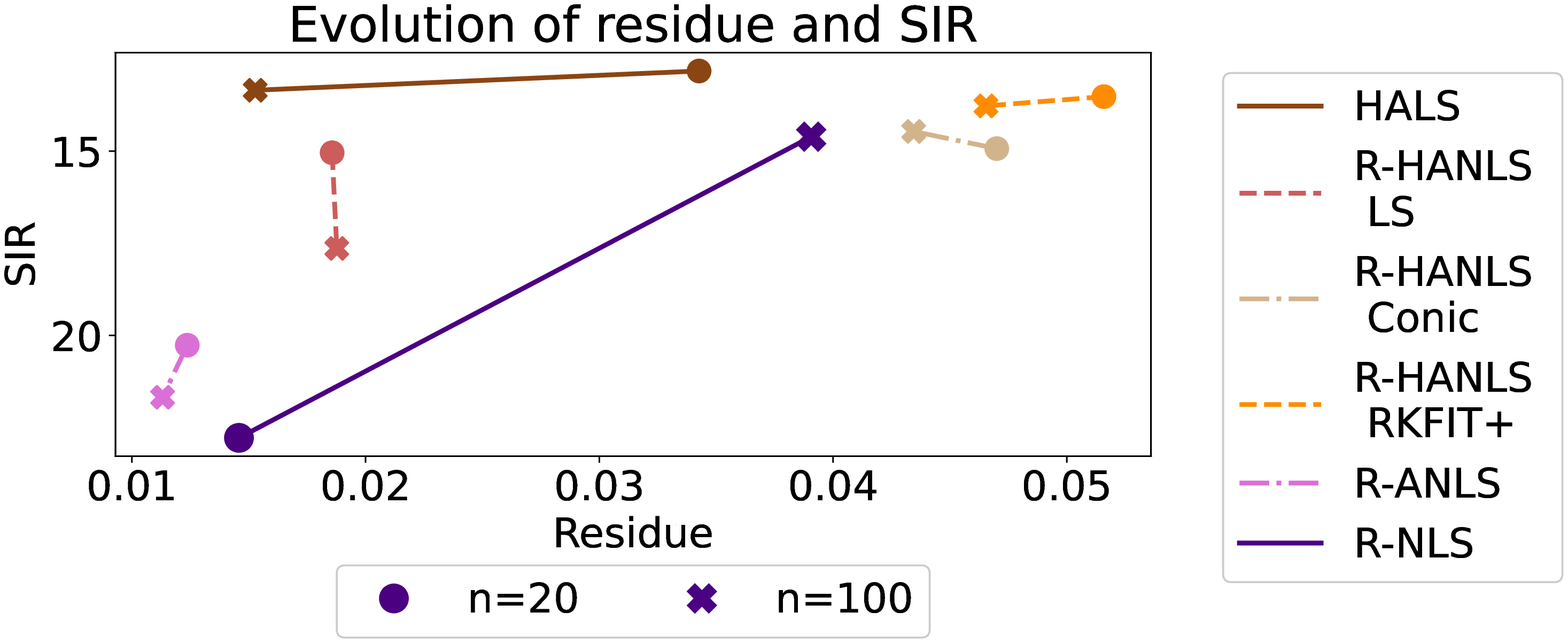}
    \caption{Average performance for varying $n$, data is noisy.}
    \label{fig:n_noise}
\vspace*{-0.3cm}
\end{figure}
\subsection{Semi-synthetic dataset}
We saw in previous sections that using rational functions in NMF when data is composed of rational functions can help significantly the algorithm, but is very sensitive to initialization. The use of rational functions is especially relevant for difficult problems, i.e. for high noise levels and when only a few observations are available. 

Let us analyze the performance of the algorithms in the semi-synthetic case, when the noise level is high (20dB) and the number of observations is low ($n=20$). This will allow us to validate whether using rational function is beneficial in such situations. 
We compared the methods to HALS as before, but also to HALS using polynomials or splines presented in~\cite{hautecoeurNeuro}. We also considered combining the R-ANLS and the R-HANLS LS methods, to try to obtain a method obtaining the same quality as R-ANLS with speed comparable to R-HANLS LS, and to have thus the best of the two algorithms. When combining these two approaches, we run one of them until the relative residue was below $10^{-2}$, and we use the result of this first method as initialization of the second method. 

Fig.~\ref{fig:refl} displays the results. 
We observe  that the R-NMF methods  obtain the smallest residues, and are thus  best to filter the noise. Among these methods, R-NLS obtains the best SIR, but it is also quite slow despite the small number of observations. R-ANLS and  the combination R-ANLS/R-HANLS obtain also good SIR  values. Note that the combination is able to obtain accuracy close to the one obtained by R-ANLS but much faster. 
The objective of combining methods is therefore met in this case. HALS using polynomials or splines also filters well the noise while HALS has more difficulties. However, all methods have difficulties to recover the original signals, as the SIR are low on average for all methods. 
Fig.~\ref{fig:refl2} shows that when a small number of signals are mixed, $r=3$,  some methods based on rational functions manage to recover a good approximation of the original signals, but when the number of original signal increases, for $r=5$ or $8$, the recovered signals do not really resemble the original ones, as illustrated in Fig.~\ref{fig:refl3}. We also observe in this figure that the signals recovered by HALS are very nonsmooth.  

On another hand, changing the degree does not influence the SIR. However, Fig.~\ref{fig:refld} shows that choosing a too low number of degrees of freedom $(d=12)$
penalizes the algorithms in terms of relative residue, especially when using polynomials or splines. 
The fact that rational functions already obtain good results for $d=12$ can be explained by the fact that rational functions are able to express a larger variety of shapes than polynomials or splines for the same degrees of freedom. However, this advantage turns into a drawback when the number of degrees of freedom is too high. Indeed,  the performances of the methods using rational functions are slightly degraded for larger degrees, because the algorithm starts to model the noise. This is the case in particular for R-HANLS LS and R-ANLS/R-HANLS. Nevertheless, the variability seems to be reduced in this case (the worst case is better than when using a lower number of degrees of freedom).  
\begin{figure}[h!]
\vspace*{-0.5cm}
    \centering
\begin{tikzpicture}
\node at (-0.05,0) { \includegraphics[width=0.49\textwidth, trim=0.2cm 0cm 0cm 0cm, clip]{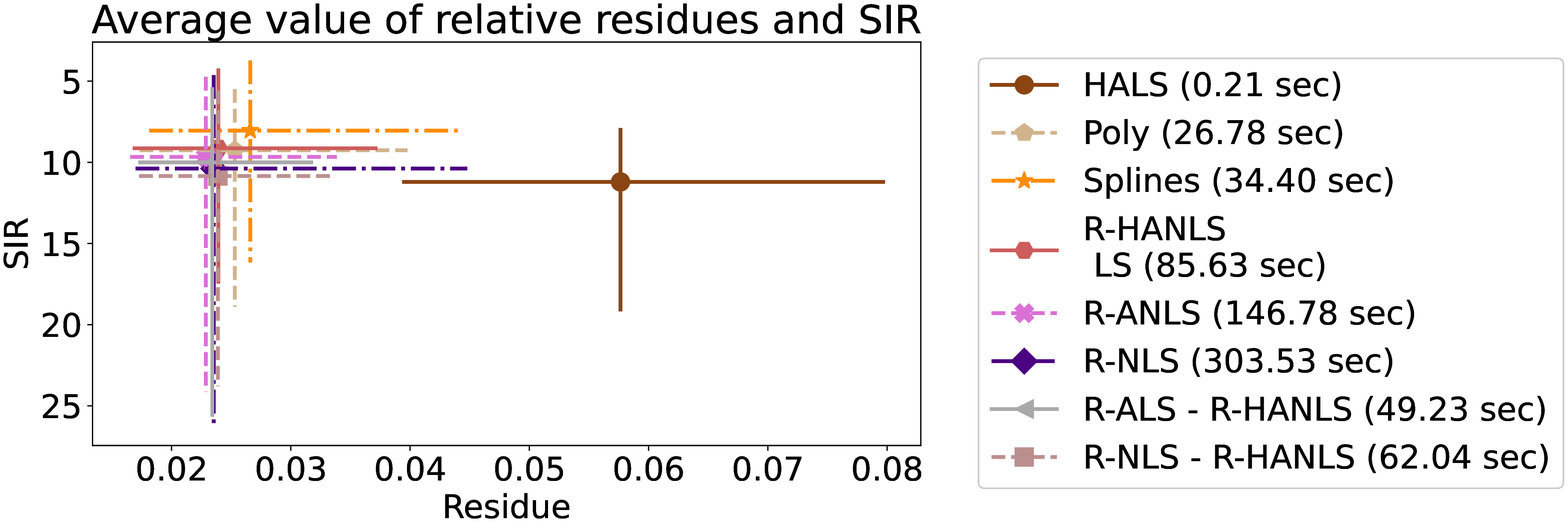}};
\draw [draw=red] (-3.4,0.75) rectangle(-3.1,0.35);

\node at (0.3,-0.3) {\includegraphics[width = 0.03\textwidth, trim= 2.5cm 3cm 32.5cm 3.5cm, clip]{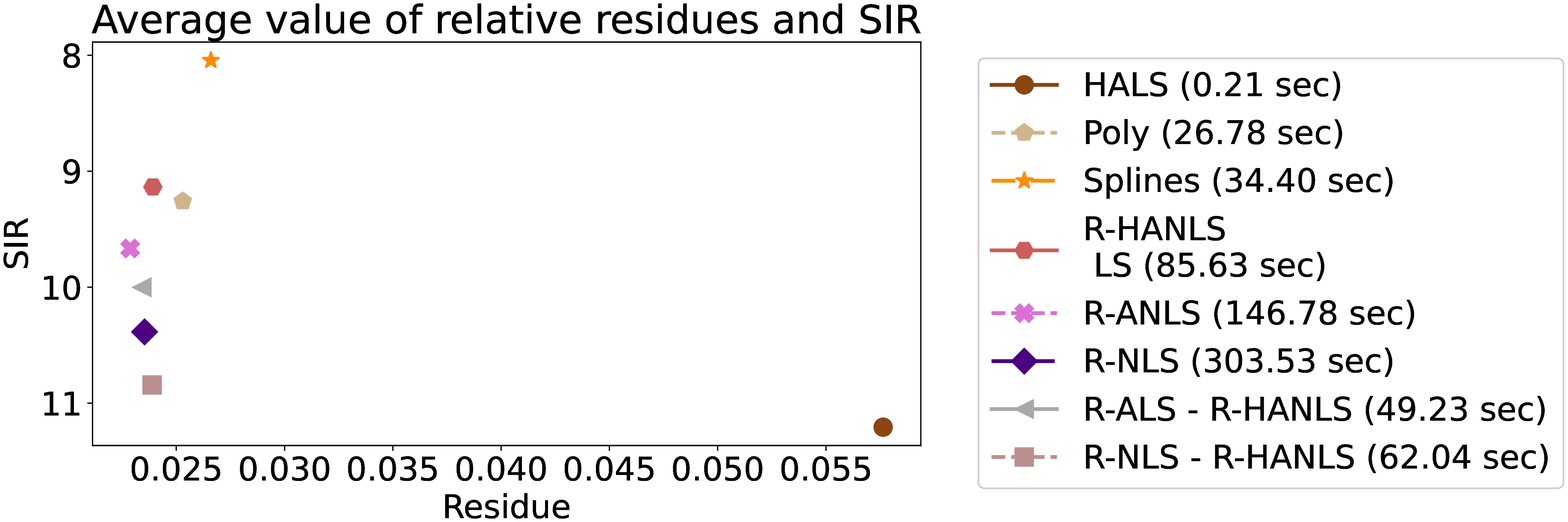}};
\draw [draw=red] (-3.08,0.55)--(0.08,-0.3);
\draw [draw=red] (0.08,0.35) rectangle(0.52,-0.95);
\end{tikzpicture}
\vspace*{-0.7cm}
    \caption{Performance on semi-synthetic dataset.}
    \label{fig:refl}
\end{figure}
\begin{figure}[h!]
    \centering
    \includegraphics[width=0.46\textwidth]{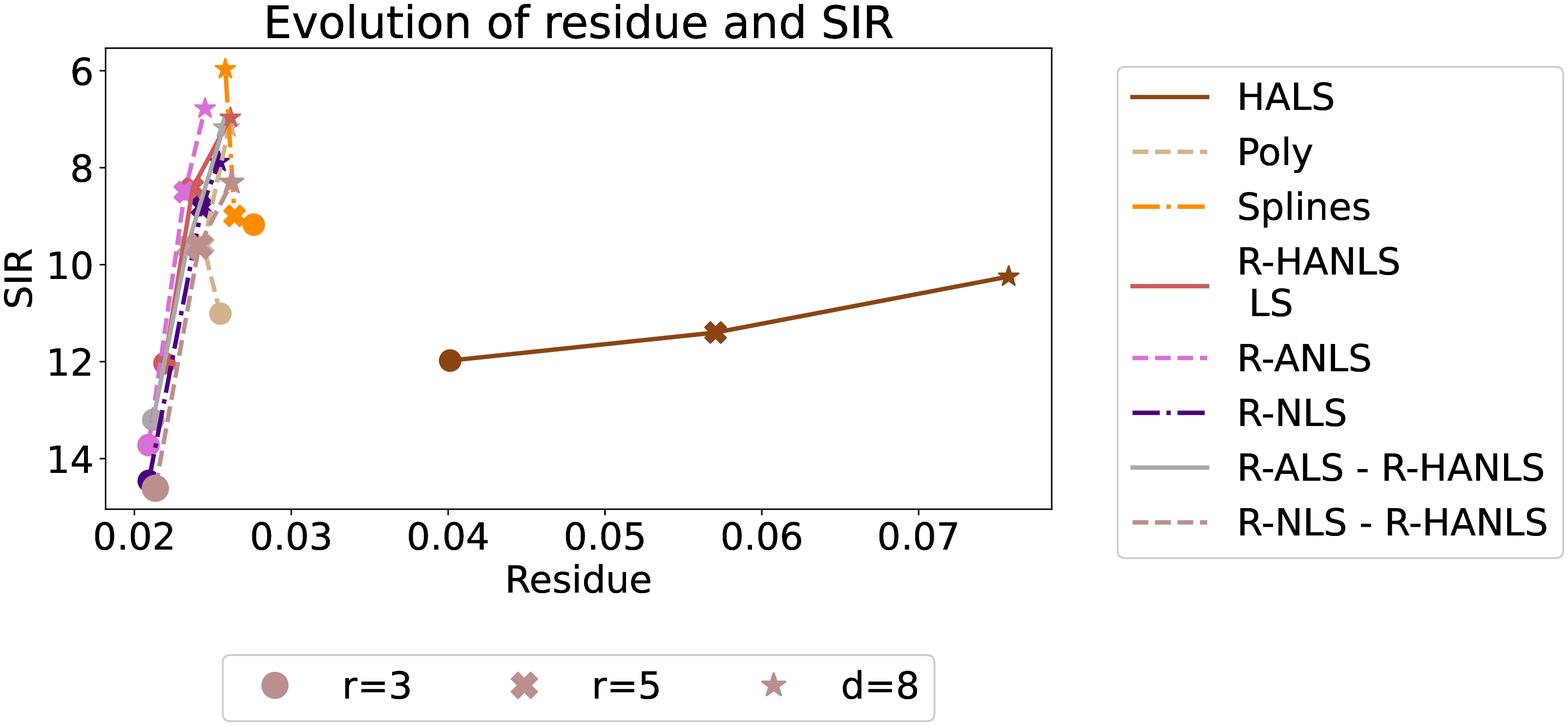}
    \vspace*{-0.2cm}
    \caption{Performance for varying rank.}
    \label{fig:refl2}
\end{figure}
\begin{figure}[h!]
\vspace*{-0.5cm}
    \centering
    \begin{tikzpicture}
    \node at (-2.2,1.6) {\includegraphics[width=0.245\textwidth]{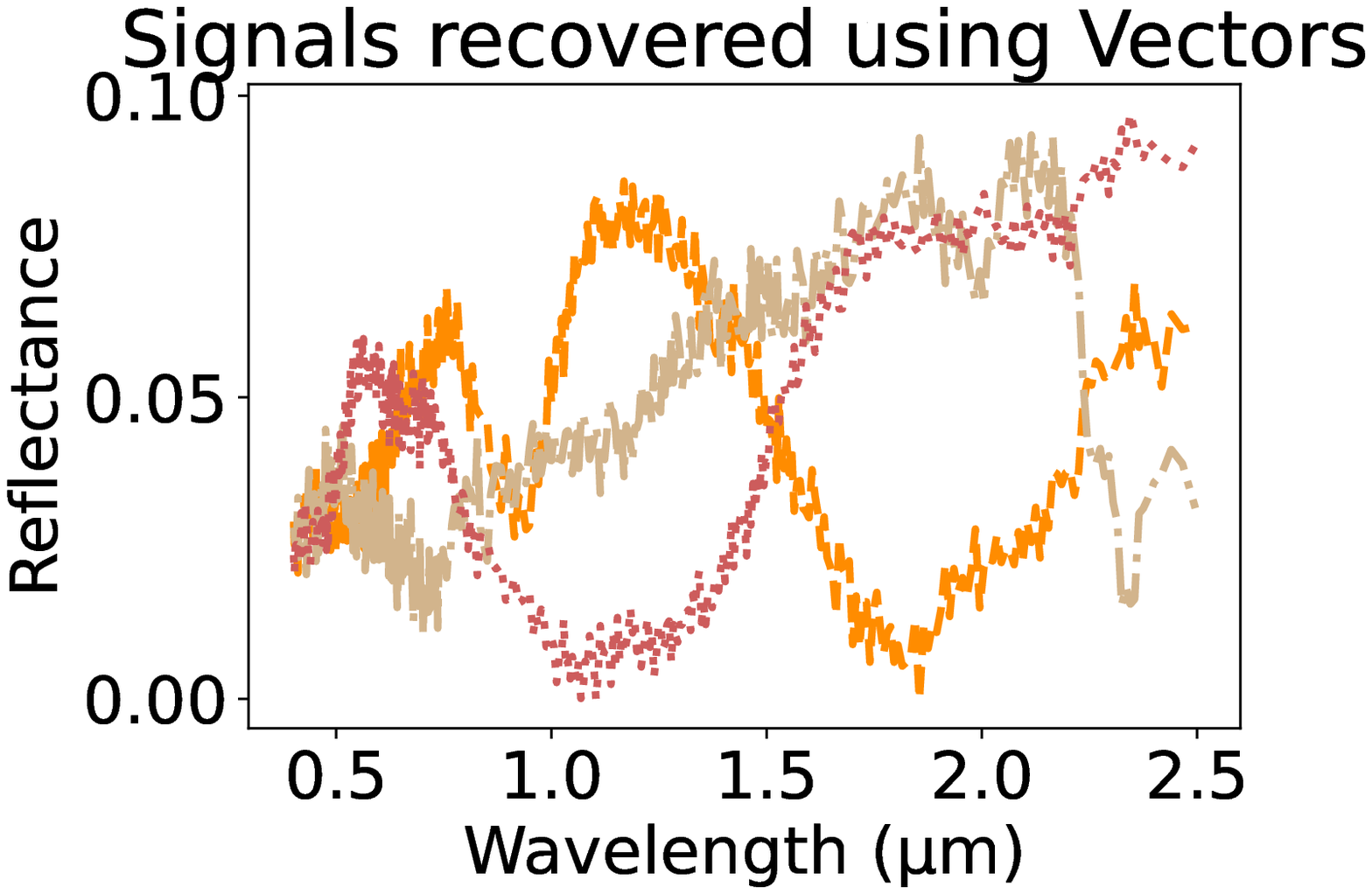}};
    \node at (2.2,1.6) {\includegraphics[width=0.245\textwidth]{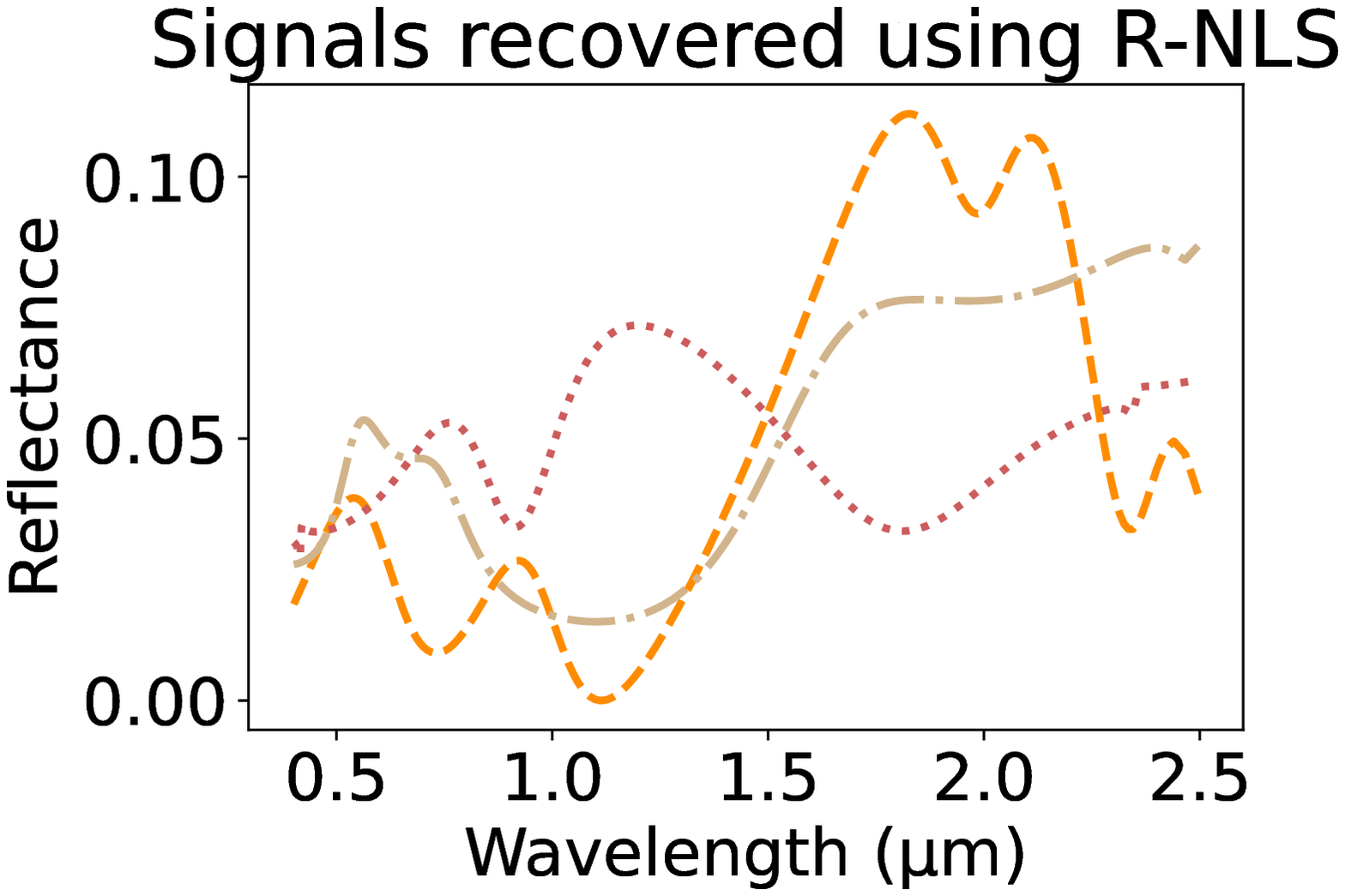}};
    \node at (-2.2,-1.6) {\includegraphics[width=0.245\textwidth]{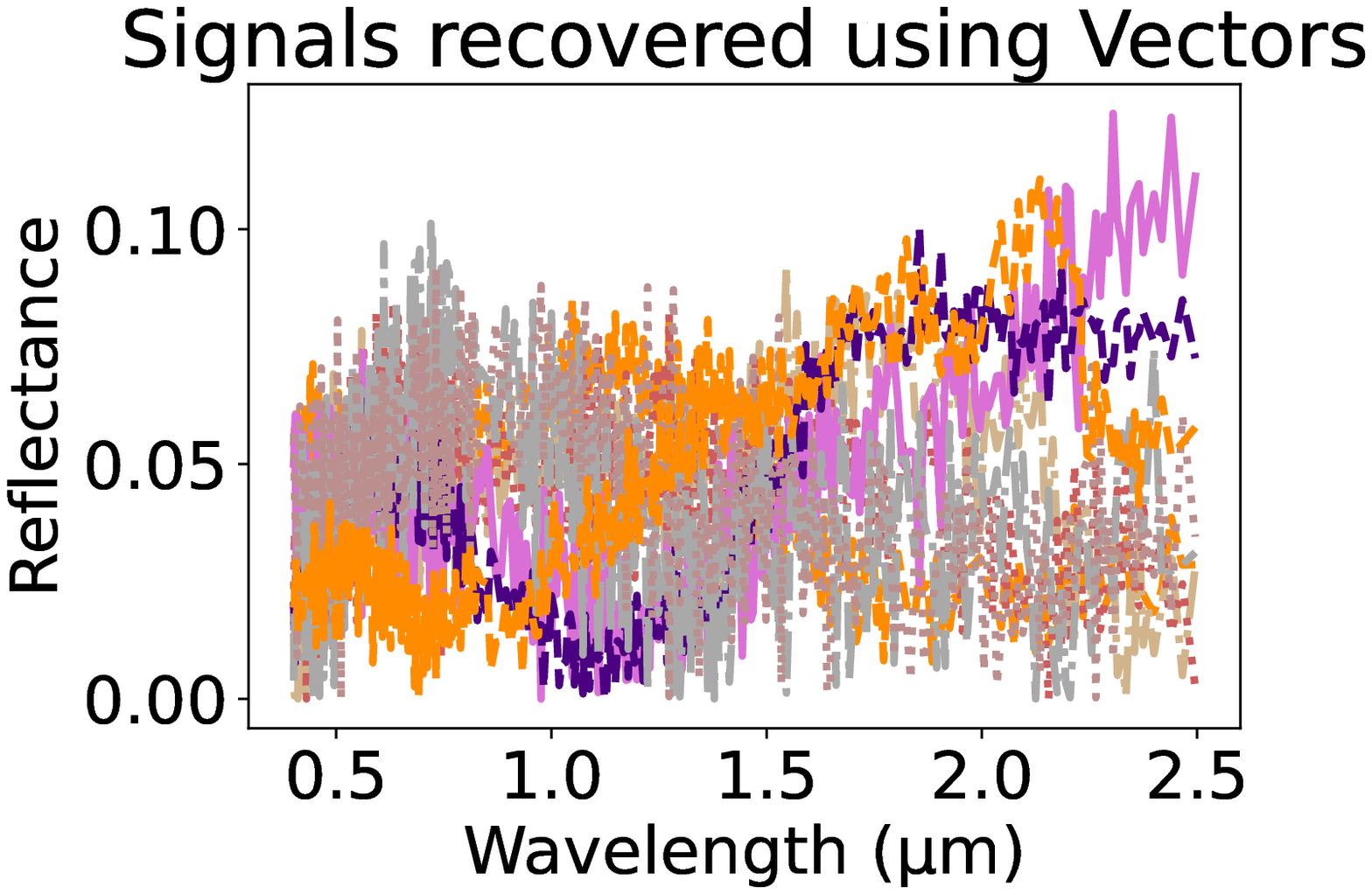}};
    \node at (2.2,-1.6) {\includegraphics[width=0.245\textwidth]{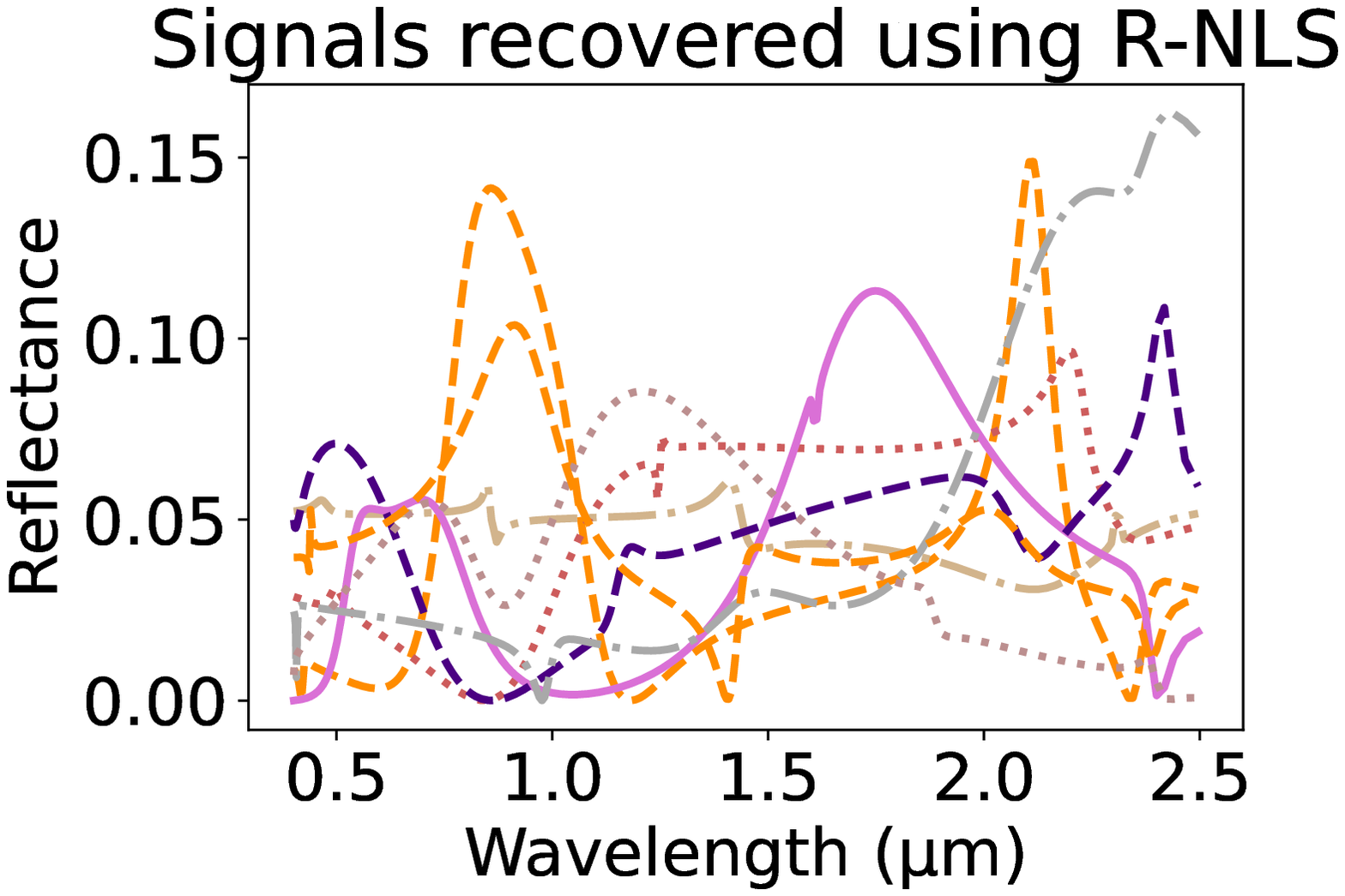}};
    \end{tikzpicture}
    \caption{Example of recovered factor $A$ for $r=3$ (up) or $8$ (down), for HALS using vectors (left) or R-NLS (right).}
    \label{fig:refl3}
    \vspace*{-0.2cm}
\end{figure}
\begin{figure}[h!]
    \centering
    \includegraphics[width=0.46\textwidth]{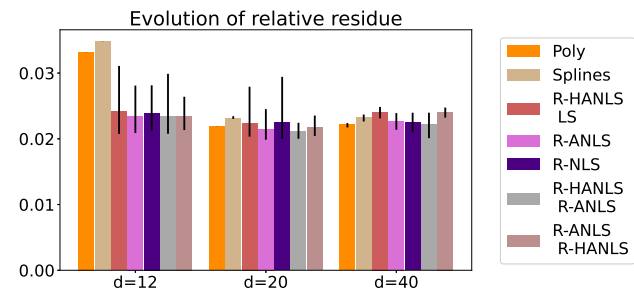}
    \caption{Performance for varying degree of freedom ($d$). }
    \label{fig:refld}
    \vspace*{-0.3cm}
\end{figure}

\vspace*{-0.3cm}
 \subsection{Using (R-)NMF for classification}
We explore the possibility of using R-NMF in a real problem: the Indian Pines classification problem. Classification is performed using the k-nearest-neighbours (KNN) algorithm with $k=5$
. A portion of 70\% of the data is used for training. 

The data is pre-processed by NMF as follows. Let $Y\in \mathbb{R}^{200\times 21025}$ be the data set, with 21025 observations of which 6307 should be classified. As the signals are spectra
, it can be assumed that they are close to polynomials, splines or rational functions. We approximate $Y$ as $AX^\top$ using NMF, where $A$ contains in its columns sampled functions (note that there is no knowledge of labels at this stage). We use the R-HANLS methods for rational functions due to the high number of observations. Then the classification is performed on $X^\top$ \hspace*{-0.1cm} instead of $Y$. The hope is that NMF filters noise in the data, while limiting the number of factors.

We also considered PCA to do the preprocessing (PCA does not have a nonnegativity constraint)\replace{, and the method from \cite{debals2015lowner} using rational functions, but without nonnegativity contraints, called Rational}{. We also tested the method of \cite{debals2015lowner} but the results were not convincing (the accuracy was always below 68\%). Perhaps the size of the dataset is too large, or imposing the degrees to be always equal is not optimal for this approach. Nevertheless, we tested the factorization with rational functions without nonnegativity constraints, using our R-HANLS algorithm, with projection onto rational functions using a least squares solver (Rational). This projection may not be ideal in this case without nonnegativity, but it gives an idea of performance. It also shows that our approach can easily be extended to other sets than the set of nonnegative rational functions.} Methods are tested 10 times over different initializations. The number of degrees of freedom is $20$, and all methods are limited to 100 seconds. The best factorization for each rank is selected using a K-fold with 5 folds on the 70\% of data constituting the training set. As a base line, we use the result of the classification on the whole data set without preprocessing. It is thus independent of the rank, and corresponds to rank $r=200$.

\begin{figure}[htb] 
    \centering
    \vspace*{-0.3cm}
    \hspace*{-0.2cm}
    \includegraphics[width = 0.5\textwidth]{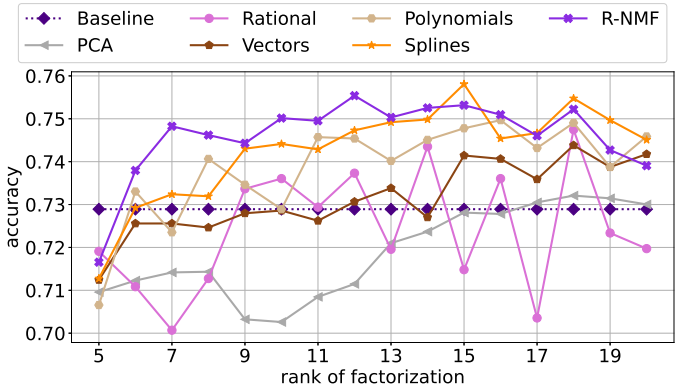}
    \vspace*{-0.6cm}
    \caption{Accuracy of classification using NMF as preprocessing with various factorization ranks.}
    \label{fig:classi}
    \vspace*{-0.15cm}
\end{figure}
Fig.~\ref{fig:classi} shows the accuracy obtained according to the factorisation rank considered during pre-processing.
It confirms the interest of R-NMF since this method obtains the best results when $r<15$. For higher rank values, NMF using splines also obtains very good results\replace{, but R-NMF stays competitive. Moreover,  it is the only method that outperforms the baseline as soon as $r=6$. It is therefore able to represent the data using a very small rank.}{, while R-NMF starts to slightly overfit.} We also see that imposing nonnegativity makes sense, since PCA and Rational which do not have this constraint obtain the worst results\replace{ and do not manage to improve the baseline. Imposing factors to be rational without imposing the nonnegativity seems to be a very bad choice in this case.}{.}

On the other hand, using standard NMF \replace{leads most of the time to accuracy very close to the accuracy obtained when using all features (the base line)}{improves the baseline only for ranks higher than $15$}, while using polynomials or splines improves accuracy compared to standard NMF, but to a lesser extent than when using rational functions
\vspace*{-0.1cm}
\subsection{Discussion}
We observed that R-NMF performs better than other NMF approaches on semi-synthetic data or real-life data. A likely explanation is that, as polynomials and splines, rational functions have less parameters than data points, and hence some form of noise averaging takes place unlike for HALS using vectors. Moreover, they generalize polynomials and splines, and are thus able to express a wider range of shapes, which allows R-NMF to recover more representative signals.
On the other hand, the presented methods to compute R-NMF do not obtain very satisfactory results when the data are actually rational functions. Indeed, even when there is no noise, these methods are not always able to recover the original signals and this despite the fact that the factorization to be recovered is unique, see Section~\ref{sec:unique}.

To explain this phenomenon,  note that for R-ANLS and R-HANLS each update is not guaranteed to be optimal, and these two methods do thus inexact BCD. But doing inexact BCD was not a problem for polynomials or splines \cite{hautecoeur2019heuri}, so this explanation is not enough. 
An other explanation is that the set of rational functions is not convex, and is not even closed for addition, so there may be many local minima in which the algorithms can get stuck, which also explain why R-NMF is very sensitive to the initialization. 

Moreover, R-NMF approaches and especially R-ANLS and R-NLS are more complex and more resources demanding than NMF using splines polynomials or vectors. 
One way of investigation to reduce the complexity of the algorithms is to consider other representations of rational functions than fractions of polynomials that could be more accurate, but for which the nonnegativity condition is not trivial, like barycentric representation \cite{filip2018rational,nakatsukasa2018aaa} or sum of fractions \cite{gustavsen1999rational}, which is left for future work. 

Furthermore, the methods presented in this paper can be extended to a wider range of rational functions where the numerator and the denominator are not imposed to be nonnegative polynomials, but can be any nonnegative function. To use least-squares based methods, a parametrization of the nonnegativity of the used functions is necessary. If an R-HANLS approach is chosen, the only necessity is that the projection exists. This means, for the Least Squares or the Alternating Least Squares projection, that a parametrization of the nonnegativity of the used functions exists. For Conic projection, a description of the nonnegativity constraint of the used functions must exist (without caring if it is the numerator or the denominator). RKFIT+ requires an operator $h'$ computing the best numerator when the denominator is fixed (possibly neglecting the nonnegativity). Finally, the LinProj projection requires the functions that are used to be linearly parametrizable, in order to keep the problem linear.
This comment highlights the many existing possibilities when performing R-NMF. 

\section{Conclusion}
We introduced R-NMF, a factorization model using nonnegative rational functions to unmix sampled signals, and presented three  approaches to solve the problem. When comparing with standard NMF or with NMF over polynomials or splines, we found that the use of rational functions can outperform existing methods, for synthetic datasets and also for a real life dataset, at the cost of an increase in computational time for large-scale data and greater sensitivity to initialization. This better reconstruction is probably due to the wider range of representation of rational functions. 

From our results, it appears that R-HANLS 
obtains on average worse results than R-ANLS. On an other hand, R-NLS is able to obtain good results on very small problems, but when the problem size increases the method slows down very strongly and has difficulties to converge. Moreover, R-NLS is resource demanding, and R-ANLS also but to a lesser degree. Therefore, we recommend to use R-NLS only for very small problems, when $n<50$ for example. For small problems, R-ANLS is accurate and not too slow (when $n<1000$). However, for larger problems, R-HANLS is more appropriate as it is much less demanding. However, when possible, it should be initialized by a few iterations of R-ANLS to improve performances. 

To take the best of R-NMF, it is necessary to continue to investigate methods of resolution, for example by combining the presented methods or by using other representation of rational functions. Nevertheless, the presented methods can be used for rational function in the broadest sense (not only for the ratio of two polynomials) under some conditions, which widens the field of possibilities for NMF.


\appendix
\section{Description and implementation of the projection methods}
\label{sec:implementation}
We describe the projection methods in more details.

\textbf{Least Squares:} we use the \texttt{least\_squares} method of python, provided with the jacobian of the cost function, with default parameters. It therefore solves the problem using trust region reflective algorithm. The algorithm is stopped when either the cost function is not enough improved anymore, or the iterates are too close from each others, or the norm of the gradient is very small. 
    
\textbf{Alternating Least Squares: } problem (\ref{eq:rkfitNum}) is as a conic problem. Indeed, using Markov-Lukacs theorem, nonnegative polynomials can be expressed using sum of squares polynomials (SOS), and SOS can be expressed using positive semidefinite matrices \cite{blekherman2012semidefinite}. Therefore, problem (\ref{eq:rkfitNum}) can be rewritten using appropriate matrices $V_{\bm \tau}(g)$ a Vandermonde-like matrix taking into account the known denominator,  
    and $R$ 
    the matrix recovering the coefficients of the polynomial from the positive semi-definite matrices. $R$ is built using Gram matrices (\cite{hautecoeurNeuro}). Let $\mathcal{S}_+^d$ be the set of positive semidefinite matrices in $\mathbb{R}^{d\times d}$. We have \vspace*{-0.2cm}
    \begin{equation} \vspace*{-0.3cm}\label{eq:conicF} \min_{\substack{(S_1,S_2) \in  \mathcal{S}_+^{\frac{d_1}{2}+1}\times \mathcal{S}_+^{\frac{d_1}{2}}}} \Bigg|\Bigg|\bm z - V_{\bm \tau}(g) R\begin{bmatrix} vec(S_1)\\vec(S_2)\end{bmatrix}\Bigg|\Bigg|^2.\end{equation}
      Problem (\ref{eq:conicF}) can be compressed using the singular value decomposition of $V_{\bm \tau}(g) = U\Sigma W^\top$. It can be proved that using $\tilde{V} = \Sigma W^\top$ 
      and $\tilde{\bm z} = U^\top \bm z$ leads to the same minimization problem, to one constant. It is solved using Mosek 9.2 \cite{mosek}. 
    The problem of finding the best denominator is solved using the same solver as for Least Squares. 
     
\textbf{Conic:} 
    A way to bypass the division difficulty is to consider the modification suggested in \cite{barrodale1970two} on which we add nonnegativity constraints: \vspace*{-0.4cm}
    \begin{equation}
    \vspace*{-0.1cm}
        \min_{h \in \mathcal{P}^{d_1,T}_+,\ g\in \mathcal{P}^{d_2,T}_+, \ g(\tau_m)=1} \bigg|\bigg|\frac{\bm z g(\bm\tau) - h(\bm\tau)}{\pB(\bm \tau)} \bigg|\bigg|^2 
        \label{eq:conic0}
        \vspace*{-0.1cm}
    \end{equation}
    where $\pB \in  \mathcal{P}^{d_2,T}_+$ is fixed so that $\pB(\tau_m)=1$. This equation is equivalent to (\ref{eq:proj_ratio}) when $g(\bm\tau) = \pB(\bm\tau)>0$. It transforms the problem into a simpler problem on polynomials. 
    
    The normalisation of $g$ is important to avoid the trivial solution $g=h=0$, and can be done without loss of generality as using $\alpha h$ and $\alpha g$ leads to the same rational function $f = h/g$. 
    Unfortunately, even with normalization, this approach leads to poor reconstruction results, even when input $\bm z$ is exactly a discretization of a nonnegative rational function. We observed that the error is often much smaller on (\ref{eq:conic0}) than on (\ref{eq:proj_ratio}). For example, suppose that $g(\tau_i)$ and $h(\tau_i)$ are very small and $\pB(\tau_i)=1$. In this case, $\frac{\bm z_i g(\tau_i)-h(\tau_i)}{\pB(\tau_i)}$ can be much smaller than $\bm z_i-\frac{h(\tau_i)}{g(\tau_i)}$. Adding a regularization term on the cost function $\lambda\| h(\bm\tau)-\pB(\bm \tau)\|^2$ with various $\lambda\geq 0$ allows to reduce the problem but not in a sufficient way. We therefore slightly modify the approach and approximate $\bm z$ by $f(\bm \tau) = \frac{h(\bm\tau)}{\pB(\bm\tau) + \delta(\bm\tau)}$, where $g\in \mathcal{P}^{d_1,T}_+, \ \pB,\delta \in \mathcal{P}^{d_2,T}_+$ and $\pB$ is fixed. So $\|\bm z - f(\bm\tau)\|^2 = $
    \begin{align} 
         \bigg\| \frac{ \bm z \pB(\bm\tau) + \bm z \delta(\bm\tau) - h(\bm\tau)}{\pB(\bm\tau)} \cdot \frac{\pB(\bm\tau)}{\pB(\bm\tau)+\delta(\bm\tau)} \bigg\|^2.
    \end{align}
    As $\delta$ and $\pB$ are nonnegative, $0<\frac{\pB(\bm \tau)}{\delta(\bm \tau) + \pB(\bm \tau)} \leq1$. The cost function of (\ref{eq:conic_2}) is thus an upper bound of the cost function of problem (\ref{eq:proj_ratio}):
    \vspace*{-0.4cm}
    \begin{equation}
        \min_{h \in \mathcal{P}^{d_1,T}_+,\delta\in \mathcal{P}^{d_2,T}_+} \bigg|\bigg|\bm z + \frac{\bm z \delta(\bm\tau) - h(\bm\tau)}{\pB(\bm\tau)}\bigg|\bigg|^2.\label{eq:conic_2}
    \end{equation}
    Solving problem (\ref{eq:conic_2}) ensures to have a rational function that leads also to a low cost in problem (\ref{eq:proj_ratio}), which was not the case when solving (\ref{eq:conic0}).
   It can be solved in a similar way as (\ref{eq:conicF}). Using appropriate matrices $V_{\bm \tau}(\pB,\bm z)$ 
   and $R$
   , we have: 
   \vspace*{-0.4cm}
    \begin{equation} \label{eq:conicF2} \min_{\substack{(S_1,S_2,D_1,D_2) \in \\ \mathcal{S}_+^{\frac{d_1}{2}+1}\times \mathcal{S}_+^{\frac{d_1}{2}}\times \mathcal{S}_+^{\frac{d_2}{2}+1}\times \mathcal{S}_+^{\frac{d_2}{2}}}} \hspace*{-0.1cm} \Bigg\|\bm z + V_{\bm \tau}(\pB, \bm z) R\begin{bmatrix} vec(S_1)\\vec(S_2) \\vec(D_1) \\vec(D_2) \end{bmatrix}\Bigg\|^2. \hspace*{-0.05cm}\end{equation}

    Problem (\ref{eq:conicF2}) can be compressed, using the singular value decomposition of $V_{\bm \tau}(\pB, \bm z) = U\Sigma W^\top$, with $\tilde{V} = \Sigma W^\top$ and $\tilde{\bm z} = U^\top \bm z$. This problem is solved using Mosek 9.2 solver. 
    
\textbf{RKFIT+: } operator $h'$ from (\ref{eq:rkfit}) can  be solved analytically using matrix $V_1$ such that $\frac{h(\bm\tau)}{\pB(\bm \tau)} = V_1 \bm h$, where $\bm h$ is the coefficient vector of $h$. 
Problem becomes:
    \begin{equation}
        \frac{h'(\pB,\delta,\bm z,\bm{\tau})}{\pB(\bm{\tau})} = V_1 \  \text{argmin}_{\bm h}  \Big\|\bm z + \frac{\bm z\delta(\bm \tau)}{\pB(\bm \tau)} - V_1 \bm h\Big\|^2.
    \end{equation}
    The solution of this problem can be expressed using  $V_1^\dagger$ the pseudo-inverse of $V_1$ as: 
    $    \frac{h'(\pB,\delta,\bm z,\bm{\tau})}{\pB(\bm \tau)} = V_1V_1^\dagger \bigg( \bm z + \frac{\bm z \delta(\bm \tau)}{\pB(\bm\tau)} \bigg) \label{eq:atilde}
    .$
    
    Similarly, we can define $V_2$ so that $\frac{\bm z \delta(\bm \tau)}{\pB(\bm\tau)} = V_2\bm \delta$, where $\bm \delta$ is the coefficient vector of $\delta$. Problem (\ref{eq:rkfit}) is then $ \label{eq:rkfit3} \min_{\delta\in \mathcal{P}^{d_2,T}_+} \| (I-V_1V_1^\dagger)(\bm z + V_2 \bm \delta) \|^2.$
    This problem can be compressed, using SVD decomposition of $ (I-V_1V_1^\dagger)V_2$: $U\Sigma W^\top $. The cost becomes $\| 
    U^\top(I-V_1V_1^\dagger)\bm z + \Sigma W^\top \bm \delta \|^2$.     The problem can then be solved using  Mosek 9.2.
    
    
\textbf{LinProj:}  this problem is solved using Mosek 9.2. This solver sometimes consider a problem as feasible when the constraint is violated by a value smaller than $10^{-6}$. To avoid this small violation to lead to a huge value of $\max_i\big(\big|\bm{z}_i - \frac{h(\tau_i)}{g(\tau_i)}\big| \big)$, $g(\bm \tau)$ is imposed to be greater than $1$. 

\normalem
\bibliographystyle{abbrv}
\bibliography{biblio}
\end{document}